\documentclass[11pt,letterpaper]{article}
\usepackage[margin=1in]{geometry}
\usepackage{amssymb,amsmath,amsthm}
\usepackage{graphicx}
\usepackage[colorlinks,citecolor=blue,linkcolor=blue,urlcolor=red,pagebackref]{hyperref}
\usepackage{subcaption}
\usepackage{booktabs}
\usepackage{color}
\usepackage{enumitem}
\usepackage{tikz}
\usepackage[noend]{algpseudocode}
\usepackage{framed}
\usepackage{thm-restate}
\usepackage{algorithm}
\usepackage{multirow}
\usepackage[T1]{fontenc}
\newtheorem{theorem}{Theorem}%[section]
\newtheorem{definition}[theorem]{Definition}%[section]
\newtheorem{corollary}[theorem]{Corollary}%[section]
\newtheorem{claim}[theorem]{\bfseries{Claim}}
\newtheorem{lemma}[theorem]{\bfseries{Lemma}}
\newtheorem{observation}[theorem]{\bfseries{Observation}}

\newcommand{\eps}{\varepsilon}
\renewcommand{\S}{{\cal S}}

\renewcommand{\leq}{\leqslant}
\renewcommand{\geq}{\geqslant}

\newbox\ProofSym
\setbox\ProofSym=\hbox{%
\unitlength=0.18ex%
\begin{picture}(10,10)
\put(0,0){\framebox(9,9){}}
\put(0,3){\framebox(6,6){}}
\end{picture}}

\begin{document}

\title{
    Approximating Partition in Near-Linear Time
}

\author{
  Lin Chen\thanks{chenlin198662@zju.edu.cn. Zhejiang University. Part of this work was done when the author was affiliated with Texas Tech University.}
  \and 
  Jiayi Lian\thanks{jiayilian@zju.edu.cn. Zhejiang University.}
  \and
  Yuchen Mao\thanks{maoyc@zju.edu.cn. Zhejiang University. Supported by National Natural Science Foundation of China [Project No. 12271477]}
  \and
  Guochuan Zhang\thanks{zgc@zju.edu.cn. Zhejiang University. Supported by National Natural Science Foundation of China [Project No. 12131003]}
}

\date{}

\maketitle
\begin{abstract}
    We propose an $\widetilde{O}(n + 1/\eps)$-time FPTAS (Fully Polynomial-Time Approximation Scheme) for the classical  Partition problem. This is the best possible (up to a polylogarithmic factor) assuming SETH (Strong Exponential Time Hypothesis) [Abboud, Bringmann, Hermelin, and Shabtay'22].  Prior to our work, the best known FPTAS for Partition runs in $\widetilde{O}(n + 1/\eps^{5/4})$ time [Deng, Jin and Mao'23, Wu and Chen'22].  Our result is obtained by solving a more general problem of weakly approximating Subset Sum.
\end{abstract}

\section{Introduction}

\paragraph{Subset Sum} Given a (multi-)set $X$ of $n$ positive integers and a target $t$, Subset Sum asks for a subset $Y \subseteq X$ with maximum $\Sigma(Y)$ that does not exceed $t$, where $\Sigma(Y)$ is the sum of the integers in $Y$. It is a fundamental problem in computer science and operations research. Although NP-Hard, it admits FPTASes (Fully Polynomial-Time Approximation Schemes).  The first FPTAS for Subset Sum was due to Ibarra and Kim~\cite{IK75} and Karp~\cite{Kar75} in the 1970s.  Since then, a great effort has been devoted to developing faster FPTASes (see Table~\ref{table:SubsetSum}). The latest one is by Kellerer, Pferschy and Speranza~\cite{KPS97}, and has an $\widetilde{O}(n+1/\eps^2)$\footnote{In the paper we use an $\widetilde{O}(\cdot)$ notation to hide polylogarithmic factors in $n$ and $\frac{1}{\eps}$.} running time. This was recently shown, by Bringmann and Nakos~\cite{BN21b}, to be the best possible (up to a polylogarithmic factor) assuming the $(\min,+)$-convolution conjecture.

\paragraph{Partition} Partition is a special case of Subset Sum where $t=\Sigma(X)/2$. It is often considered as one of ``the easiest NP-hard Problems'' and yet has many applications in scheduling \cite{CL91}, minimization of circuit sizes and cryptography \cite{MH78}, and game theory \cite{Hay02}.  Despite the fact that Partition can be reduced to Subset Sum, algorithms specific to Partition have been developed in the hope of solving it faster than Subset Sum. See Table~\ref{table:Partition}. Mucha, W\k{e}grzycki and W\l{}odarczyk~\cite{MWW19} gave the first subquadratic-time FPTAS for Partition that runs in $\widetilde{O}(n+1/\eps^{5/3})$ time.  This result, together with the quadratic lower bound for Subset Sum, implies that Partition is easier than Subset Sum at least in terms of polynomial-time approximation schemes. On the negative side, the conditional lower bound of $\mathrm{poly}(n)/\eps^{1-o(1)}$ from Abboud, Bringmann, Hermelin, and Shabtay~\cite{ABHS22} implies that a running time of $O((n + \frac{1}{\eps})^{1-o(1)})$ is impossible assuming SETH. This naturally raises the following question.
\begin{center}
    \emph{Can Partition be approximated within $O(n + \frac{1}{\eps})$ time?}
\end{center}
There is a line of work trying to resolve this question. Bringmann and Nakos~\cite{BN21b} gave a deterministic FPTAS that runs in $\widetilde{O}(n+1/\eps^{3/2})$. The current best running time is $\widetilde{O}(n+1/\eps^{5/4})$ which was obtained independently by Deng, Jin, and Mao~\cite{DJM23} and by Wu and Chen~\cite{WC22}. The question is still away from being settled.

\paragraph{Weakly Approximating Subset Sum} The study of weak approximation schemes for Subset Sum was initiated by Mucha, W\k{e}grzycki and W\l{}odarczyk~\cite{MWW19}. Weak approximation lies between strong\footnote{From now on, we say that the standard approximation is strong, in order to distinguish it from weak approximation.} approximation for Subset Sum and Partition in the sense that any weak approximation scheme for Subset Sum implies a strong approximation scheme for Partition. Consider a Subset Sum instance $(X, t)$. Given any $\eps > 0$, a weak approximation scheme finds a subset $Y$ of $X$ such that 
\[
    (1 - \eps)\Sigma(Y^*) \leq \Sigma(Y) \leq (1+\eps)t,
\]
where $Y^*$ is the optimal solution, that is, a weak approximation scheme allows a solution to slightly break the constraint (also known as resource augmentation).
Mucha, W\k{e}grzycki and W\l{}odarczyk ~\cite{MWW19} proposed an $\widetilde{O}(n+1/\eps^{5/3})$-time weak approximation scheme.  It was later improved to $\widetilde{O}(n+1/\eps^{3/2})$~\cite{BN21b, WC22}.  The lower bound is $(n + \frac{1}{\eps})^{1-o(1)}$, which is the same as that for Partition.

\subsection{Our Results}

\begin{theorem}\label{thm:subsetsum}
    There is an $\widetilde{O}(n + \frac{1}{\eps})$-time randomized weak approximation scheme for Subset Sum, which succeeds with probability at least $1 - (\frac{n}{\eps})^{-O(1)}$.
\end{theorem}

Theorem~\ref{thm:subsetsum} immediately implies the following theorem since any weak approximation scheme for Subset Sum is a strong approximation scheme for Partition. 

\begin{theorem}\label{thm:partition}
     There is an $\widetilde{O}(n + \frac{1}{\eps})$-time randomized FPTAS for Partition, which succeeds with probability at least $1 - (\frac{n}{\eps})^{-O(1)}$.
\end{theorem}

Both of the above two results match (up to a polylogarithmic factor) the lower bound of $(n + \frac{1}{\eps})^{1-o(1)}$.  To our best knowledge, Partition is the first NP-hard problem that admits an FPTAS that is near-linear in both $n$ and $1/\eps$.  We also remark that our weak approximation scheme for Subset Sum generalizes Bringmann's $\widetilde{O}(n + t)$-time exact algorithm for Subset Sum~\cite{Bri17} in the sense that when $\eps : = \frac{1}{2t}$, the weak approximation scheme becomes an exact algorithm with $\widetilde{O}(n + t)$ running time.

To attain Theorem~\ref{thm:subsetsum}, we utilize an additive combinatorics result that is different from the ones previously used in Subset Sum and Knapsack~\cite{GM91, MWW19, BN20, BW21, WC22, DJM23, CLMZ23}. We believe that it may be of independent interest.

\begin{table}[!ht]
    \centering
    \caption{Polynomial-time approximation schemes for Subset Sum. Reference with a star (*) indicates that the algorithm was designed for Knapsack and works for Subset Sum. Symbol (\dag) means that it is a randomized approximation scheme.}
    \begin{tabular}{cc}
        \toprule 
        \quad Running Time\qquad & Reference \\
        \toprule 
        \multicolumn{2}{c}{Strong approximation scheme for Subset Sum}\\
        \hline \specialrule{0em}{0pt}{2pt}
        $O(n/\eps^2)$ & *Ibarra and Kim~\cite{IK75}, *Karp~\cite{Kar75}\\
        $O(n/\eps)$ & Gens and Levner~\cite{GL78, GL79}\\
        $O(n + 1/\eps^4)$ & *Lawler~\cite{Law79}\\
        $O(n + 1/\eps^3)$ & Gens and Levner~\cite{GL94}\\
        $\widetilde{O}(n + 1/\eps^2)$ & Kellerer, Mansini, Pferschy and Speranza\cite{KMPS03}\\
        \hline
        C.L.B. $(n+1/\eps)^{2-o(1)}$ & Bringmann and Nakos~\cite{BN21b}\\
        %\multicolumn{2}{c}{Conditional Lower Bound: $(n+1/\eps)^{2-o(1)}$~\cite{BN21b}}\\
        \toprule 
        \multicolumn{2}{c}{Weak approximation scheme for Subset Sum}\\
        \hline \specialrule{0em}{0pt}{2pt}
        $\widetilde{O}(n+1/\eps^{5/3})$ &\dag Mucha, W\k{e}grzycki and W\l{}odarczyk~\cite{MWW19}\\
        $\widetilde{O}(n+1/\eps^{3/2})$ & \dag Bringmann and Nakos~\cite{BN21b}, Wu and Chen~\cite{WC22}\\
        $\widetilde{O}(n + 1/\eps)$ & \dag This Paper \\
        \hline
        C.L.B. $\mathrm{poly}(n)/\eps^{1-o(1)}$ & Abboud, Bringmann, Hermelin and Shabtay~\cite{ABHS22}\\
        %\multicolumn{2}{c}{Conditional Lower Bound: $\mathrm{poly}(n)/\eps^{1-o(1)}$~\cite{ABHS22}}\\
    \bottomrule 
    \end{tabular}
    \label{table:SubsetSum}
\end{table}

\begin{table}[!ht]
    \centering
    \caption{Polynomial-time approximation schemes for Partition. Symbol (\dag) means that it is a randomized approximation scheme.}
    \begin{tabular}{cc}
        \toprule 
        \quad Running Time\qquad & Reference \\
        \hline \specialrule{0em}{0pt}{2pt}
        $\widetilde{O}(n + 1/\eps^2)$ & Gens and Levner~\cite{GL80} \\
        $\widetilde{O}(n+1/\eps^{5/3})$ &\dag Mucha, W\k{e}grzycki and W\l{}odarczyk~\cite{MWW19}\\
        $\widetilde{O}(n+1/\eps^{3/2})$ & Bringmann and Nakos~\cite{BN21b}\\
        $\widetilde{O}(n + 1/\eps^{5/4})$ & Deng, Jin and Mao~\cite{DJM23}, Wu and Chen~\cite{WC22} \\
        $\widetilde{O}(n + 1/\eps)$ &\dag This Paper\\
        \hline
        C.L.B. $\mathrm{poly}(n)/\eps^{1-o(1)}$& Abboud, Bringmann, Hermelin and Shabtay~\cite{ABHS22}\\
        %\multicolumn{2}{c}{Conditional Lower Bound: $\mathrm{poly}(n)/\eps^{1-o(1)}$~\cite{ABHS22}}\\
        \bottomrule 
    \end{tabular}
    \label{table:Partition}
\end{table}

\subsection{Technical Overview} 
Let $X = \{x_1, \ldots, x_n\}$. Solving Subset Sum can be reduced to computing the sumset $\{x_1, 0\} + \cdots + \{x_n, 0\}$. Our main technique is an approach for efficiently approximating (a major part of) the sumset of integer sets, which can be seen as a combination of a deterministic sparse convolution~\cite{BFN22} (see Lemma~\ref{lem:sparse-fft}) and an additive combinatorics result from Szemer{\'e}di and Vu~\cite{SV05} (see Theorem~\ref{lem:add-comb}). We briefly explain the idea below.

Suppose that we are to compute the sumset of $A_1, \ldots, A_{\ell}$. We compute it in a tree-like manner.  $A_1, \ldots, A_{\ell}$ forms the bottom level of the tree. We compute the next level by taking the sumset of every two nodes in the bottom level. That is, we compute $A_1 + A_2, A_3 + A_4, \ldots, A_{\ell-1} + A_{\ell}$.  Let $B_i = A_{2i-1} + A_{2i}$ for $i=1,\ldots,\ell/2$.  If $B_1, \ldots, B_{\ell/2}$ have small total sizes, then we can compute all of them efficiently via sparse convolution, and proceed to the next round. When $B_1, \ldots, B_{\ell/2}$ has a large total size, we cannot afford to compute them. Instead, we utilize an additive combinatorics result from Szemer{\'e}di and Vu~\cite{SV05} to show that  $A_1 + \cdots + A_{\ell}$ has a long arithmetic progression. We will further show that this arithmetic progression can be extended to a long sequence with small differences between consecutive terms. Given the existence of such a long sequence, we can directly compute a set that nicely approximates (a major part of) $A_1 + \cdots + A_{\ell}$.

Besides sparse convolution and the additive combinatorics result, the following techniques are also essential for obtaining a near-linear running time.
\begin{itemize}
    \item Recall that we compute a level only when it has a small total size. Therefore, we should estimate the size of a level without actually computing it. The estimation can be done by utilizing the sparse convolution algorithm~\cite{BFN22}.
    
    \item The additive combinatorics result~\cite{SV05} only guarantees the existence of a long arithmetic progression, and is non-constructive. Therefore, we can only obtain a good approximation of solution values, but may not be able to recover solutions efficiently. A similar issue also arises in the algorithm for dense Subset Sum~\cite{BW21}.  In Section~\ref{sec:algo-solution}, we show the issue can be fixed in an approximate sense.
    
    \item Rounding is necessary every time when we compute a new level. If the tree has too many levels, then the error caused by rounding would accumulate and become far more than what is acceptable.  To ensure a small number of tree levels, we use the color-coding technique, which was proposed by Bringmann~\cite{Bri17} to design exact algorithms for Subset Sum. For technical reasons, the color-coding technique is slightly modified to ensure additional properties.
\end{itemize}

Our algorithm, and in particular, the usage of arithmetic progression, is inspired by the dense Subset Sum algorithm by Bringmann and Wellnitz~\cite{BW21}, and Galil and Margalit~\cite{GM91}, but we obtain the arithmetic progression in a way that is different from these two works. Bringmann and Wellnitz~\cite{BW21}, and Galil and Margalit~\cite{GM91} got the arithmetic progression via an additive combinatorics result saying that if an integer set $A$ has $\widetilde{\Omega}(\sqrt{u})$ distinct integers, where $u$ is the maximum integer in $A$, then the collection of subset sums of $A$ must contain an arithmetic progression of length ${\Omega}(u)$ {(see, e.g., \cite{Sar94})}. In contrast, we use a different result from additive combinatorics, which states that if we have $\ell$ sets $A_1, \ldots, A_\ell$ of size at least $k$  and $\ell k \geq \Omega(u)$, then $A_1 + \cdots + A_\ell$ must contain an arithmetic progression of length ${\Omega}(u)$.

\subsection{Further Related Work}
Subset Sum is a special case of Knapsack. There is a long line of research on approximation schemes for Knapsack, e.g., \cite{IK75, Kar75, Law79, KP04, Rhee15, Cha18, Jin19, DJM23}. There is a conditional lower bound of $(n + 1/\eps )^{2-o(1)}$ based on the (min, +)-convolution hypothesis\cite{CMWW19, KPS17}. Very recently, \cite{CLMZ23a} and \cite{Mao23} establish an $\widetilde{O}(n+1/\eps^2)$-time FPTAS for Knapsack, independently.

In addition to approximation algorithms, exact pseudopolynomial-time algorithms for Knapsack and Subset Sum have also received extensive studies in recent years, e.g., \cite{Bri17, AT19, PRW21, BC23, CLMZ23, Bri23a, Jin23}. 
It is worth mentioning that there is an algorithm~\cite{BN20} with running time $\widetilde{O}({k^{4/3}})$ for Subset Sum where $k$ is the number of subset sums that are smaller than $t$, which shares a similar flavor to our algorithm in the sense that it combines sparse convolution and Ruzsa’s triangle inequality for sumset estimation (which is a different result in additive combinatorics).

Our algorithm utilizes sparse convolution algorithms, see, e.g.,~\cite{CH02, AR15, CL15a, Nak20a, BFN21, BN21a, BFN22}. We also adopt results from additive combinatorics, see, e.g., \cite{Alo87, Sar89, Sar94, SV05, SV06}. Additive combinatorics results have been exploited extensively in recent years on Knapsack and Subset Sum, see, e.g., \cite{GM91, MWW19, BN20, BW21, WC22, DJM23, CLMZ23}.

\subsection{Paper Organization} In Section~\ref{sec:pre}, we introduce some necessary terminology and tools, and show that it suffices to solve a reduced problem. In Section~\ref{sec:alg-determine}, we present a near-linear-time algorithm for approximating optimal values. We further show how to recover an approximate solution in near-linear time in Section~\ref{sec:algo-solution}. Section~\ref{sec:conclude} concludes the paper. All the omitted proofs can be found in the appendices.

\section{Preliminary}\label{sec:pre}
\subsection{Notation and Definitions}
Throughout this paper, we assume that $\eps > 0$ is  sufficiently small. We also assume that $\frac{1}{\eps}$ is an integer by adjusting $\eps$ by an $O(1)$ factor. 
All logarithms ($\log$) in this paper are base $2$.

Let $w,v$ be two real numbers.  We use $[w,v]$ to denote the set of integers between $w$ and $u$.  That is, $[w,v] = \{z \in \mathbb{Z} : w\leq z\leq v\}$.  Let $Y$ be a nonempty set of integers. We write $Y \cap [w,v]$ as $Y[w,v]$. We denote the minimum and maximum elements of $Y$ by $\min(Y)$ and $\max(Y)$, respectively.  We write $\sum_{y \in Y}y$ as $\Sigma(Y)$. We refer to the number of elements in $Y$ as the size of $Y$, denoted by $|Y|$. Through this paper, we use both the terms ``set'' and ``multi-set''. Only when the term ``multi-set'' is explicitly used do we allow duplicate elements. Unless otherwise stated, a subset of a multi-set is a multi-set.

Let $(X, t)$ be an arbitrary Subset Sum instance. We assume that $x \leq t$ for any $x \in X$ since the integers greater than $t$ can be safely removed from $X$.    We define $\S_X$ to be the set of all subset sums of $X$. That is, $\S_X = \{\Sigma(Y) : Y\subseteq X\}$. Subset Sum actually asks for the maximum element of $\S_X[0,t]$. 

We consider the more general {problem} of approximating $\S_X$.
\begin{definition}\label{def:approx}
    Let $S$ be a set of integers. Let $w,v$ be two real numbers. We say a set $\widetilde{S}$ approximates $S[w,v]$ with additive error $\delta$ if 
    \begin{enumerate}[label={\normalfont(\roman*)}]
        \item for any $s \in S[w,v]$, there is $\tilde{s} \in \widetilde{S}$ with $s - \delta\leq \tilde{s} \leq s + \delta$, and

        \item for any $\tilde{s} \in \widetilde{S}$, there is $s \in S$ with $\tilde{s} -\delta \leq s \leq \tilde{s} + \delta$.
    \end{enumerate}
\end{definition}

Our algorithm has two phases. In the first phase, it computes a set $\widetilde{S}$ that approximates $\S_X[0,t]$ with additive error $\eps t$. Let $\tilde{s}$ be the maximum element of $\widetilde{S}[0,(1 + \eps) t]$. Clearly, $\tilde{s} \leq (1 + \eps) t$.  Definition~\ref{def:approx}(i) implies that $\tilde{s} \geq \Sigma(Y^*) - \eps t$, and Definition~\ref{def:approx}(ii) implies that there exists $Y \subseteq X$ such that $\tilde{s} - \eps t \leq \Sigma(Y) \leq \tilde{s} + \eps t$.  Therefore, 
\[
         (1- 4\eps)\Sigma(Y^*) \leq  \Sigma(Y^*) - 2\eps t \leq \tilde{s} - \eps t \leq \Sigma(Y) \leq \tilde{s} + \eps t \leq t + 2\eps t
\]
The first inequality is due to the fact that 
we can assume the optimal objective value $\Sigma(Y^*) \geq t/2$\footnote{If $\Sigma(Y^*) < t/2$, every interger in $X$ must be less than $t/2$. Then it implies that $Y^* = X$ because otherwise, we can improve $Y^*$ by selecting one more integer in $X$. Such an instance can be solved trivially in $O(n)$ time.}. By adjusting $\eps$ by a factor of $4$, we have that 
\[
    (1- \eps)\Sigma(Y^*) \leq \Sigma(Y) \leq (1 +\eps) t.
\]
The second phase of the algorithm recovers such a $Y$ from $\tilde{s}$.

\subsection{Sumset}
Let $(X_1,X_2)$ be a partition of $X$.  We have $\S_X = \S_{X_1} + \S_{X_2}$, where the sum of two sets is defined by the following.
\[
    A + B = \{a + b : a \in A, b\in B\}.
\]
The set $A + B$ is called the sumset of $A$ and $B$.  It is well-known that the sumset of two integer sets can be computed via the classical convolution algorithm based on Fast Fourier Transformation (FFT). 
\begin{lemma}\label{lem:fft}
    Let $u$ be a positive integer. Let $A$ and $B$ be two subsets of $[0,u]$. We can compute their sumset $A + B$ in $O(u\log u)$ time.
\end{lemma}

When the sumset has a small size, sparse convolution algorithms whose running time is linear in the output size may be used. Although there are faster randomized algorithms, to ease the analysis, we use the following deterministic algorithm due to Bringmann, Fischer, and Nakos~\cite{BFN22}.
\begin{lemma}[\normalfont \cite{BFN22}]\label{lem:sparse-fft}
    Let $u$ be a positive integer. Let $A$ and $B$ be two subsets of $[0,u]$. We can compute their sumset $A + B$  in $O(|A + B| \log^5 u\,\mathrm{polyloglog}\,u)$ time.
\end{lemma}

Lemma~\ref{lem:fft} immediately implies a few simple approaches for approximating the sumset of two sets.

%\begin{lemma}\label{lem:approx-two-fft}
%Let $u$ be a positive integer. Let $A$ and $B$ be two subsets of $[0,u]$ with total size $k$. For any $\eps < 1$, in $O(k + \frac{1}{\eps}\log \frac{1}{\eps})$ time, we can compute a set $S$ of size $O(\frac{1}{\eps})$ that approximates $(A+B)[0, u]$ with additive error $\eps u$.
%\end{lemma}
%\begin{proof}
%    We first round the elements of $A$ and $B$ down to the nearest multiple of $\eps u/2$, and delete duplicate elements in each set. This step takes $O(k + \frac{1}{\eps})$ time and incurs an additive error of at most $\eps u/2$. $A$ and $B$ can be viewed as subsets of $[0, \frac{2}{\eps}]$ (after dividing their elements by $\eps u/2$). Therefore, we  can compute $S := A + B$ in $O(\frac{1}{\eps}\log\frac{1}{\eps})$ time via Lemma~\ref{lem:fft}.  It is straightforward that $S[0,u]$ is of size $O(\frac{1}{\eps})$ and approximates $(A+B)[0,u]$ with additive error $\eps u$.  
%\end{proof}
\begin{lemma}\label{lem:approx-k-fft}
    Let $u$ be a positive integer. Let $A_1,\ldots, A_{\ell}$ be subsets of $[0,u]$ with total size $k$. For any $\eps < 1$, in $O(k + \frac{\ell^2}{\eps}\log \frac{\ell}{\eps})$ time, we can compute a set $S$ of size $O(\frac{1}{\eps})$ that approximates $(A_1+\cdots+A_{\ell})[0, u]$ with additive error $\eps u$.
\end{lemma}
\begin{proof}
    We first round the elements of all $A_i$'s down to the nearest multiple of $\eps u$, and delete duplicate elements in each $A_i$. This step takes $O(k + \frac{\ell}{\eps})$ time and incurs an additive error of at most $\eps u$ for each $A_i$. The we can compute $(A_1+\cdots+A_{\ell})[0, u]$ via Lemma~\ref{lem:fft}, and the total running time is $O(\frac{\ell}{\eps}\log \frac{1}{\eps})$. The total additive errors is $\ell \eps u$.  By adjusing $\eps$ by a factor of $\ell$, the total additive error becomes $\eps u$ and the running becomes $O(k + \frac{\ell^2}{\eps}\log \frac{\ell}{\eps})$.

    %We compute in a binary-tree-like manner. The leaves are $A_1, \ldots, A_\ell$. We say the leaves are of level $0$. Given level $i$, we construct level $i+1$ by merging every two nodes in level $i$ via Lemma~\ref{lem:approx-two-fft}.  We repeat until we get a level with a single node.  The tree has $O(\ell)$ nodes whose total size is $O(k + \frac{\ell}{\eps})$ (every internal node is of size $O(\frac{1}{\eps})$ and the leaves have a total size of $k$). Therefore, the total running time is $O(k + \frac{\ell}{\eps} + \frac{\ell}{\eps}\log \frac{1}{\eps})$.  We incur an additive error of $\eps u$ in each node, so the total additive error is $\ell\eps u$.  By adjusting $\eps$ by a factor of $\ell$, we can reduce the total additive error to $\eps u$, while increase the total running time to $O(k+ \frac{\ell^2}{\eps}\log \frac{\ell}{\eps})$.
\end{proof}

\subsection{Problem Reduction}
We can reduce our task to the following problem $\mathrm{RP}(\beta)$ where  $\beta\in[1,\frac{1}{\eps}]$ be an integer.

\begin{definition}[The Reduced Problem $\mathrm{RP}(\beta)$]
    Let $X$ be a multi-set of integers from $[\frac{1}{\eps}, \frac{2}{\eps}]$ such that $\Sigma(X) \geq \frac{4\beta}{\eps}$.  (i) Compute a set $\widetilde{S}$ of size $\widetilde{O}(\frac{1}{\eps})$ that approximates $\S_X[\frac{\beta}{\eps}, \frac{2\beta}{\eps}]$ with additive error $\delta= {O}(\beta\,\, \mathrm{polylog}\,\frac{n}{\eps})$. (ii) Given any $\tilde{s} \in \widetilde{S}$, recover a subset $Y \subseteq X$ such that $\tilde{s}-\delta \leq \Sigma(Y) \leq \tilde{s}+ \delta$.
\end{definition}
\begin{restatable}{lemma}{lemreduce}
\label{lem:reduce}
    There is an $\widetilde{O}(n + \frac{1}{\eps})$-time weak approximation scheme for Subset Sum if, for any $\beta\in[1,\frac{1}{\eps}]$, the reduced problem $\mathrm{RP(\beta)}$ can be solved in $\widetilde{O}(n + \frac{1}{\eps})$ time. 
    %\begin{quote}
        %Let $\beta\in[1,\frac{1}{\eps}]$ be an integer. Let $X$ be a multi-set of integers from $[\frac{1}{\eps}, \frac{2}{\eps}]$ such that $\Sigma(X) \geq \frac{4\beta}{\eps}$.  (i) Compute a set $\widetilde{S}$ of size $\widetilde{O}(\frac{1}{\eps})$ that approximates $\S_X[\frac{\beta}{\eps}, \frac{2\beta}{\eps}]$ with additive error $\delta= {O}(\beta\,\, \mathrm{polylog}\,\frac{n}{\eps})$. (ii) Given any $\tilde{s} \in \widetilde{S}$, recover a subset $Y \subseteq X$ such that $\tilde{s}-\delta \leq \Sigma(Y) \leq \tilde{s}+ \delta$.
        %Let $\alpha\in [1, \frac{1}{\eps}]$ and $\beta \in [1, \frac{1}{\alpha\eps}]$ be two integers. Let $X$ be a multi-set of integers from $[\frac{1}{\eps}, \frac{2}{\eps}]$.  Assume that $\Sigma(X) \geq \frac{4\beta}{\eps}$.  (i) Compute a set $\widetilde{S}$ of size $\widetilde{O}(\frac{1}{\eps})$ that approximates $\S_X[\frac{\beta}{\eps}, \frac{2\beta}{\eps}]$ with additive error $\delta= \widetilde{O}(\frac{1}{\alpha\eps})$. (ii) Given any $\tilde{s} \in \widetilde{S}$, recover a subset $Y \subseteq X$ such that $(1-\delta)\tilde{s} \leq Y \leq (1 + \delta)\tilde{s}$.
    %\end{quote}
\end{restatable}
The proof of the following lemma is deferred to Appendix~\ref{sec:reduce}. Basically, we can preprocess the instance to get rid of tiny integers less than $\eps t$. Then we scale the instance by $\eps^2 t$. After that, all integers in $X$ are within $[\frac{1}{\eps}, \frac{1}{\eps^2}]$. Then we partition $X$ into $\log\frac{1}{\eps}$ groups so that the integers within the same group differ by a factor of at most $2$. We deal with each group separately and merge their results by Lemma~\ref{lem:approx-k-fft}. By further scaling, we can assume that $x \in [\frac{1}{\eps}, \frac{2}{\eps}]$ for each group.

In what follows, we present an algorithm for part (i) of the reduced problem $\mathrm{RP}(\beta)$ in Section~\ref{sec:alg-determine}, and an algorithm for part (ii) in Section~\ref{sec:algo-solution}.

\section{Approximating the Set of Subset Sums}\label{sec:alg-determine}
Lemma~\ref{lem:approx-k-fft} already implies an $O(\frac{n^2}{\eps}\log \frac{n}{\eps})$-time algorithm for approximating $\S_X[0,t]$: recall that $n = |X|$, partition $X$ into $n$ singletons, and merge via Lemma~\ref{lem:approx-k-fft}. Recall that the algorithm works in a tree-like manner. The inefficiency of this algorithm is mainly due to the large node number of the tree and the large cost at each node. To improve the running time, our algorithm utilizes the following two techniques.
\begin{enumerate}[label={(\roman*)}]
    \item We reduce the number of tree nodes via the two-layer color-coding from Bringmann~\cite{Bri17}. 

    \item Before a new level of the tree is computed, we estimate the total size of the nodes in this level. Only when the total size is small do we compute this level, and we can compute efficiently via output-sensitive algorithms (Lemma~\ref{lem:sparse-fft}). When the total size is large, we show that a good approximation can be immediately obtained via additive combinatorics tools.
\end{enumerate}

\subsection{Color Coding}
Recall that our goal is to approximate $\S_X[\frac{\beta}{\eps}, \frac{2\beta}{\eps}]$ and that $x \in [\frac{1}{\eps}, \frac{2}{\eps}]$ for any $x \in X$. 
For any $Y\subseteq X$ with $\Sigma(Y) \in \S_X[\frac{\beta}{\eps}, \frac{2\beta}{\eps}]$, we have $|Y| \leq 2\beta$. Color-coding is a technique for dealing with such a situation where the subsets under consideration are small.

Let $Y$ be any subset of $X$ with $|Y| \leq k$.  Let $X_1, \ldots, X_{k^2}$ be a random partition of $X$. By standard balls and bins analysis, with probability at least $\frac{1}{4}$, $|Y \cap X_i| \leq 1$ for all $i \in [1, k^2]$, which implies that $\Sigma(Y) \in (X_1 \cup \{0\}) + \cdots + (X_{k^2}\cup \{0\})$.  This probability can be boosted to $1-q$, if we repeat the above procedure for $\lceil\log_{4/3} q \rceil$ times and take the union of the resulting sumsets.

Bringmann~\cite{Bri17} proposed a two-layer color coding technique that reduces the number of subsets in the partition by roughly a factor of $k$.  Let $q$ be the target error probability. The first layer of color coding randomly partitions $X$ into roughly $k/\log(k/q)$ subsets. With probability at least $1-q/2$, each subset contains at most $6\log(k/q)$ elements from $Y$. The second layer randomly partitions each of the subsets obtained in this first layer into $36\log^2(k/q)$ subsets.  The second-layer partition will be repeated for $\lceil \log \frac{k}{q} \rceil$ times in order to boost the success probability to $1-q/2$. See Algorithm~\ref{alg:orignal-color-coding} for details.

\begin{algorithm}
\caption{$\mathtt{ColorCoding}(X, k, q)$~\cite{Bri17}}
\label{alg:orignal-color-coding}
    \begin{algorithmic}[1]
    \Statex \textbf{Input:} A multi-set $X$ of integers, a positive integer $k$, and a target error probability $q$
    \Statex \textbf{Output:} $r$ partitions $\{X^j_{1,1}, \ldots, X^j_{1,g}, \ldots, X^j_{m,1}, \ldots, X^j_{m, g}\}_{j \in [1, r]}$ of $X$
    \State $m:= k/\log(k/q)$ rounded up to be next power of $2$\;
    \State $g:= 36\log^2(k/q)$ rounded up to the next power of $2$\;
    \State $r:=\lceil \log \frac{k}{q} \rceil$\;
    \State Randomly partition $X$ into $m$ subsets $X_1, \ldots, X_m$\;
    \For{$j =1,\ldots,r$}
        \For{$i = 1,\ldots,m$}
           \State Randomly partition $X_i$ into $g$ subsets $X^j_{i,1}, \ldots, X^j_{i,g}$
        \EndFor
    \EndFor
  \State \Return the $r$ partitions $\{X^j_{1,1}, \ldots, X^j_{1,g}, \ldots, X^j_{m,1},$ $ \ldots, X^j_{m, g}\}_{j \in [1, r]}$ of $X$. 
  \end{algorithmic}
\end{algorithm}

\begin{lemma}[\cite{Bri17}]\label{lem:original-color-coding}
    Let $m,g,r$ be defined as that in Algorithm~\ref{alg:orignal-color-coding}. Let $\{X^j_{1,1}, \ldots, X^j_{1,g}, \ldots, X^j_{m,1},$ $ \ldots, X^j_{m, g}\}_{j \in [1, r]}$ be the partitions of $X$ returned by $\mathtt{ColorCoding}(X, k, q)$ in Algorithm~\ref{alg:orignal-color-coding}.  For $j \in [1,r]$, let $S^j_i = (X^j_{i,1}\cup \{0\}) + \cdots + (X^j_{i,g} \cup \{0\})$. For any subset $Y\subseteq X$ with $|Y| \leq k$, with probability at least $1 - q$,
    \[
        \Sigma(Y) \in \bigcup_{j=1}^r S^j_1 +  \cdots + \bigcup_{j=1}^r S^j_m.
    \]
\end{lemma}

By adjusting the error probability by a factor of $\frac{2\beta}{\eps}$, we can have, with probability at least $1-q$, that $\bigcup_{j=1}^r S^j_1 +  \cdots + \bigcup_{j=1}^r S^j_m$ contains all the elements of $\S_X[\frac{\beta}{\eps}, \frac{2\beta}{\eps}]$.

For technical reasons, we need an extra property that for any partition of $X$ obtained from coloring-coding, the sum of the maximum elements of the subsets is large. That is, we need, for all $j$,
\[
    \max(X^j_{1,1}) + \cdots + \max(X^j_{1,g}) + \cdots +\max(X^j_{m,1}) + \cdots + \max(X^j_{m, g})  \geq \frac{4\beta}{\eps},
\]
where the maximum of an empty set is defined to be $0$.  We claim that this property can be assured with a slight modification of the color coding algorithm. Details will be provided in Appendix~\ref{app:color-coding}.

\begin{restatable}{lemma}{lemcolorcoding}
\label{lem:mod-color-coding}
    Let $m$ be $4\beta/\log \frac{4\beta^2}{\eps q^*}$ rounded up to next power of $2$, let $g$ be $36\log^2\frac{4\beta^2}{\eps q^*}$ rounded up to next power of $2$, and let $r := \lceil \log \frac{4\beta^2}{\eps q^*} \rceil$. In $\widetilde{O}(n + \frac{1}{\eps})$ time, we can obtain $r$ partitions $\{X^j_{1,1}, \ldots, X^j_{m, g}\}_{j \in [1, r]}$ of $X$ such that the following is true.  For $i\in [1, m]$ and $j \in [1,r]$, let $S^j_i = (X^j_{i,1}\cup \{0\}) + \cdots + (X^j_{i,g} \cup \{0\})$. With probability at least $1 - q^*$,
    \[
        \S_X[\frac{\beta}{\eps}, \frac{2\beta}{\eps}] = \left(\bigcup_{j=1}^r S^j_1 +  \cdots + \bigcup_{j=1}^r S^j_{m}\right)[\frac{\beta}{\eps}, \frac{2\beta}{\eps}].
    \]
    Moreover, for every $j\in [1,r]$,
    \[
        \max(X^j_{1,1}) + \cdots + \max(X^j_{1,g}) + \cdots +\max(X^j_{m,1}) + \cdots + \max(X^j_{m, g}) \geq \frac{4\beta}{\eps}. 
    \]   
\end{restatable}
The proof of the lemma is deferred to Appendix~\ref{app:color-coding}.

Throughout the rest of the paper, we fix $q^*: = (\frac{n}{\eps})^{-O(1)}$, $m$ to be $4\beta/\log \frac{4\beta^2}{\eps q^*}$ rounded up to next power of $2$,  $g$ to be $36\log^2\frac{4\beta^2}{\eps q^*}$ rounded up to next power of $2$, and $r := \lceil \log \frac{4\beta^2}{\eps q^*} \rceil$. 

\subsection{Sparse and Dense Tree Levels}\label{subsec:tree-like-computing}
Let $\{X^j_{1,1}, \ldots, X^j_{m, g}\}_{j \in [1, r]}$ be the $r$ partitions of $X$ obtained via Lemma~\ref{lem:mod-color-coding}. To simplify the notation, we assume that every subset in the $r$ partitions of $X$ is a set (rather than a multi-set), and that every subset contains $0$. This assumption is without loss of generality, because when we define $S^j_i$ in Lemma~\ref{lem:mod-color-coding}, we always add $0$ to the subsets of $X$ . To approximate $\S_X[\frac{\beta}{\eps}, \frac{2\beta}{\eps}]$, it suffices to compute 
\(
    S = \bigcup_{j=1}^r S^j_1 + \cdots + \bigcup_{j=1}^r S^j_{m},
\)
where $S^j_i$ is defined in Lemma~\ref{lem:mod-color-coding}.

The procedure for computing $S$ can be viewed as a tree.  The level $0$ of the tree has $mg$ leaves that represent the $mg$ subsets of $X$ in the $j$th partition. Given level $i$, we compute level $i+1$ by taking the sumset of every two nodes in level $i$.  Therefore, in level $\log g$, we will obtain $S^j_1, \ldots, S^j_m$. The computation from level $0$ to level $g$ will be repeated for $r$ times, so we can get $S^j_1, \ldots, S^j_m$ for all $j$. Then we take union and get a new level $\log g$ consisting of $\cup_j S^j_1, \ldots, \cup_j S^j_m$. Then we continue to compute level $1 + \log g, 2 + \log g, \ldots$ until we get a single node (set) in level $\log mg$.

\begin{observation}\label{obs:level-property}
    Consider some level $h$ of the tree.  It has $\ell = \frac{mg}{2^h}$ nodes, say $A_1, \ldots, A_{\ell}$.  For each $i$, $A_i$ is a subset of $[0, \frac{2^{h+1}}{\eps}]$ and $0 \in A_i$. Moreover, $\sum_{i=1}^{\ell} \max(A_i) \geq \frac{4\beta}{\eps}$.
\end{observation}

Given a level $h$, if we compute the sumset of two sets via standard FFT (Lemma~\ref{lem:fft}), it would require $O(\frac{mg}{2^h} \cdot \frac{2^{h}}{\eps} \log \frac{2^{h}}{\eps}) = \widetilde{O}(\frac{\beta}{\eps})$ time to obtain the next level. This is already too much as $\beta$ can be as large as $\frac{1}{\eps}$.

Our algorithm computes a new level only if the nodes in this level have a total size of roughly $\widetilde{O}(\frac{1}{\eps})$. We say a level is sparse in this case. A sparse level can computed in nearly $\widetilde{O}(\frac{1}{\eps})$ via the output-sensitive algorithm for sumset whose running time is linear in the output size (Lemma~\ref{lem:sparse-fft}). The only exception is level $0$: in $O(n)$ time, we can determine whether it is sparse or not, and we will try to compute level 1 only if level $0$ is sparse. Recall that there are $1 + \log mg$ levels and that the first $1+ \log g$ levels may be repeated but only for $r$ times. Both $1 + \log mg$ and $r$ are logarithmic in $n$, $\frac{1}{\eps}$ and $\beta$. So the total running time for computing $S$ will be $\widetilde{O}(n + \frac{1}{\eps})$, if all levels are sparse.  For dense levels, we cannot afford to compute them. Instead, we use additive combinatorics tools to show that if some level is dense, $\S_X$ must have a long sequence with a small difference between consecutive terms. Then $\S_X[\frac{\beta}{\eps}, \frac{2\beta}{\eps}]$ can be immediately approximated via the following lemma.

\begin{lemma}\label{lem:easy-approx}
    Let $\delta$ be a positive integer. Suppose that $S$ has a sequence $a_1 < \ldots < a_k$ such that  $a_i - a_{i-1} \leq \delta$ for any $i \in [2, k]$.  Then for any real numbers $w$ and $v$ such that $a_1 \leq w$ and $v \leq a_k$, the following set approximates $S[w, v]$ with additive error $\delta$.
    \[
        \{w + j\delta : j \in [0, \frac{v - w}{\delta}]\}
    \]
\end{lemma}
\begin{proof}
    For any $s\in S[w,v]$, there must be some $j \in [0, \frac{v - w}{\delta}]$ such that
    \[
        w + j\delta \leq s \leq w + (j+1)\delta.
    \]
    The first condition of Definition~\ref{def:approx} is satisfied. In the other direction, for any $\tilde{s}\in\{w + j\delta : j \in [0, \frac{v - w}{\delta}]\}$, there must be some $i\in[1,k-1]$ such that
    \(
        a_i \leq \tilde{s} \leq a_{i+1}\leq a_i+\delta.
    \) 
    The second condition is also satisfied.  So $\{ w+j\delta:j\in[0,\frac{v-w}{\delta}]\}$ approximates $S[w,v]$ with additive error $\delta$.
\end{proof}

\subsection{Density and Arithmetic Progressions}
We formalize the concepts of ``dense'' and ``sparse'' and derive some properties resulting from them. The core of our technique is an additive combinatorics result from Szemer{\'e}di and Vu~\cite{SV05}, which basically states that given many large-sized sets of integers, their sumset must have a long arithmetic progression.

\begin{definition}
    An arithmetic progression of length $k$ and common difference $\Delta$ is a set of $k$ integers $a_1 <  \ldots < a_k$ such that the differences of consecutive terms are all $\Delta$.
\end{definition}

\begin{theorem}[Corollary 5.2~\cite{SV05}]\label{lem:add-comb}
    For any fixed integer $d$, there are positive constants $c_1$ and $c_2$ depending on $d$ such that the following holds.  Let $A_1, \ldots, A_{\ell}$ be subsets of $[1,u]$ of size $k$.  If $\ell^d k \geq c_1u$, then $A_1 + \cdots + A_\ell$ contains an arithmetic progression of length at least $c_2\ell k^{1/d}$.
\end{theorem}

\begin{corollary}\label{coro:add-comb}
    There exists a sufficiently large constant $c$ such that the following holds.  Let $u$ be a positive integer. Let $A_1, \ldots, A_{\ell}$ be subsets of $[0,u]$ of size at least $k$. If $\ell (k-1) \geq cu$, then $A_1+ \cdots + A_\ell$ contains an arithmetic progression of length at least $u$.
\end{corollary}
\begin{proof}
    Let $c_1$ and $c_2$ be the two constants for $d=1$ in Theorem~\ref{lem:add-comb}.  Assume that $c \geq c_1$ and that $cc_2 \geq 1$.  For $i\in [1, \ell]$, let $A^+_i = A_i \setminus \{0\}$.  $A^+_i$ is a subset of $[1,u]$ with size at least $k - 1$.  By Theorem~\ref{lem:add-comb}, $A^+_1 + \cdots + A^+_\ell$ contains an arithmetic progression of length at least $c_2\ell (k-1) \geq c_2c u \geq  u$. This arithmetic progression also appears in $A_1 + \cdots + A_\ell$ since $A^+_1 + \cdots + A^+_\ell \subseteq A_1 + \cdots + A_\ell$.
\end{proof}
Throughout the rest of the paper, $c$ denotes the constant in the above corollary. 

An arithmetic progression of length $u$ is not long enough for our purpose.  To produce a longer sequence, we need a collection of integer sets to be $\gamma$-dense.  

\begin{definition}\label{def:dense}
    Let $u$ and $\gamma$ be positive integers.  Let $A_1, \ldots, A_{\ell}$ be subsets of $[0,u]$.  The collection $\{A_1, \ldots, A_{\ell}\}$ is $\gamma$-dense if for some $k \in [2, u+1]$, at least $\frac{c\gamma u}{k-1}$ sets from this collection have size at least $k$.  We say that $\{A_1, \ldots, A_{\ell}\}$ is $\gamma$-sparse if it is not $\gamma$-dense.
\end{definition}

Through the next two lemmas, we will prove that, when a collection is $\gamma$-dense, we can use it to produce a long sequence with a small difference between consecutive elements. Basically, we use a small fraction of the collection to produce an arithmetic progression of length $u$, and use the rest of the collection to extend this progression to a longer sequence. We first show how to extend an arithmetic progression.

\begin{lemma}\label{lem:ap-extend}
    Let $B$ be a set of positive integers that contains an arithmetic progression $b_1 < \ldots < b_k$ with common difference $\Delta$. Let $u$ be a positive integer. Let $A_1, \ldots, A_{\ell}$ be subsets of $[0,u]$ containing $0$.  Let $\eta = \max(A_1) + \cdots + \max(A_\ell)$.  Let $S = B + A_1 + \cdots + A_\ell$.  If $b_k - b_1 \geq u - 1$, then $S$ contains a sequence $b_1 = s_1 < \cdots < s_{k'} = b_k + \eta$ such that $s_i - s_{i-1} \leq \Delta$ for $i \in [2, k']$.
\end{lemma}
\begin{proof}
    We prove by induction on $j$.  When $j = 0$, the lemma is obviously true by letting $k' = k$ and $s_i = b_i$ for $i\in [1,k]$. Suppose that the lemma is true for some $0\leq j  \leq \ell -1$. We prove that it is also true for $j+1$.  Let $S_j = B + A_1 + \cdots + A_j$ and $\eta_j = \max(A_1) + \cdots + \max(A_j)$. Let $S_{j+1}= S_j + A_{j+1}$ and $\eta_{j+1} = \eta_j + \max(A_{j+1})$.  By inductive hypothesis, $S_{j}$ contains a sequence $(s_1, \cdots, s_{k'})$ with $s_1 = b_1$, $s_{k'} = b_k + \eta_j$, and $s_i - s_{i-1} \leq \Delta$ for $i \in [2, k']$.  Note that $S_{j} \subseteq S_{j+1}$ since $0 \in A_{j+1}$, so $S_{j+1}$ also contains the sequence $(s_1, \ldots, s_{k'})$.  Let $a^*$ be the maximum element in $A_{j+1}$.  The sequence $(s_1 + a^*, \ldots, s_{k'} + a^*)$ also belongs to $S_{j + 1}$. Note that $s_1 + a^* \leq b_1 + u \leq b_k + 1 \leq s_{k'} + 1$ and that $s_{k'} + a^* = b_k + \eta_j + a^* = b_k + \eta_{j+1}$.  Therefore, merging the two sequences by taking union and delete duplicates yields a sequence $z_1, \ldots, z_{k''}$ in $S_{j+1}$ with $z_1 = b_1$, $z_{k''} = b_k + \eta_{j+1}$, and $z_i - z_{i-1} \leq \Delta$ for any $i \in [2,k'']$.
\end{proof}

Now we are ready to prove that if $A_1, \ldots, A_{\ell}$ is $\gamma$-dense, their sumset must contain a long sequence with small differences between consecutive terms.

\begin{lemma}\label{lem:dense2longseq}
    Let $u$ be a positive integer. Let $A_1, \ldots, A_{\ell}$ be subsets of $[0,u]$ with  $0 \in A_i$ for every $i\in[1,\ell]$. Let $S = A_1 + \cdots + A_{\ell}$. Let $\eta = \max(A_1) + \cdots + \max(A_\ell)$. If $\{A_1, \ldots, A_{\ell}\}$ is $\gamma$-dense, then $S$ contains a sequence $s_1, \ldots, s_k$ such that $s_1 \leq \frac{2\eta}{\gamma}$, $s_k \geq \frac{(\gamma-2)\eta}{\gamma}$, and $s_{i} - s_{i-1} \leq \frac{4\ell}{\gamma}$ for any $i \in [2,k]$.
\end{lemma}
\begin{proof}
    If $u = 1$, the lemma trivially holds because $[0,\eta] \subseteq S$.  Assume that $u \geq 2$.  By definition~\ref{def:dense}, for some $k \in [2, u+1]$, at least $\ell' = \frac{c\gamma u}{k-1}$ sets $A_i$'s have $|A_i| \geq k$.  From these $A_i$'s, we select the $\lceil \frac{\ell'}{\gamma} \rceil$ ones with smallest maximum elements.  Let $I$ be the set of indices of the selected $A_i$'s.   By the way, we select $A_i$'s, we have
    \[
         \sum_{i\in I}\max(A_i) \leq \frac{\lceil \frac{\ell'}{\gamma} \rceil}{\ell'}\sum_{i=1}^\ell \max(A_i) \leq \frac{2}{\gamma}\sum_{i=1}^\ell \max(A_i) = \frac{2\eta}{\gamma}.
    \]
    The second inequality is due to that $\frac{\ell'}{\gamma} \geq c \geq 1$, which follows by the fact that $\ell' = \frac{c\gamma u}{k-1}$ and that $k \leq u+1$.

    It is easy to verify that the collection $\{A_i\}_{i\in I}$ of the selected $A_i$'s satisfies the condition of Corollary~\ref{coro:add-comb}, so $\sum_{i \in I}A_i$ contains an arithmetic progression $(a_1, \ldots, a_u)$.  The common difference  
    \[
        \Delta = \frac{a_u - a_1}{u-1} \leq \frac{2a_u}{u} \leq \frac{2}{u}\cdot \sum_{i\in I}\max(A_i) \leq \frac{4\eta}{\gamma u} \leq \frac{4\ell}{\gamma}.
    \]
    The last inequality is due to that $\eta = \sum_{i=1}^\ell \max(A_i) \leq \ell u$. Moreover, $a_u - a_1 \geq u-1$.  Now consider $S = \sum_{i=1}^\ell A_i = \sum_{i\in I}A_i + \sum_{i\notin I}A_i$.  By Lemma~\ref{lem:ap-extend}, $S$ contains a sequence $s_1 ,\ldots, s_k$ such $s_1 = a_1$, $s_k = a_u + \sum_{i\notin I}\max(A_i)$, and $s_i - s_{i-1} \leq \Delta \leq \frac{4\ell}{\gamma}$. Note that
    \[
        s_1 = a_1 < a_u \leq \sum_{i \in I}\max(A_i) \leq \frac{2\eta}{\gamma}
    \]
    and
    \[
        s_k = a_u + \sum_{i\notin I}\max(A_i) \geq \sum_{i\notin I}\max(A_i) = \sum_{i=1}^\ell \max(A_i) - \sum_{i \in I}\max(A_i) \geq (1 - \frac{2}{\gamma})\eta.\qedhere
    \]
\end{proof}

The following lemma shows that a $\gamma$-sparse collection of integer sets has a small total size.

\begin{lemma}\label{lem:sparse-small-size}
    Let $A_1, \ldots, A_{\ell}$ be subsets of $[0,u]$. If $\{A_1, \ldots, A_{\ell}\}$ is $\gamma$-sparse, then 
    \[
        \sum_{i=1}^{\ell}|A_i| \leq \ell +  c\gamma u (1 + \log u).
    \]
\end{lemma}
\begin{proof}
    Note that $0 \leq |A_i| \leq u+1$ for any $i$.  For $k \in [0,u+2]$, let $\ell_k$ be the number of $A_i$'s with $|A_i| \geq k$.  Since $\{A_i\}_{i \in I}$ is $\gamma$-sparse, $\ell_k < \frac{c\gamma u}{k-1}$ for any $k \geq 2$. Then 
    \[
        \sum_{i=1}^{\ell}|A_i| = \sum_{k=1}^{u+1} k(\ell_k - \ell_{k+1}) = \sum_{k=1}^{u+1}\ell_k \leq \ell_1 + \sum_{k=2}^{u+1} \frac{c\gamma u}{k-1} 
        = \ell_1 + \sum_{k=1}^{u} \frac{c\gamma u}{k} \leq \ell +  c\gamma u (1 + \log u).\qedhere
    \]
\end{proof}

\subsection{Estimating the Density}
Recall that we compute a new level only when it is sparse, so we should estimate the density of a level without actually computing it. We first give an algorithm for estimating the size of the sumset of two integer sets. The following lemma is implied by~\cite{BFN22}.  We defer the proof to appendix~\ref{app:est-sumset-size}.

\begin{restatable}{lemma}{lemestsumsetsize}\label{lem:est-sumset-size}
    Let $u$ be a positive integer. Let $A$ and $B$ be two subsets of $[0,u]$. Let $k$ be any positive integer. We can determine whether $|A+B| \geq k$ or not in $O(k\log^6 u\,\mathrm{polyloglog}\,u)$ time.
\end{restatable}

With the above lemma, we can estimate the density of a level without computing it.

\begin{lemma}\label{lem:est-density}
   Let $u, \gamma$ be two positive integers.  Let $A_1, \ldots, A_\ell$ be non-empty subsets of $[0, u]$. For $i \in [1, \ell/2]$, let $B_i = A_{2i-1} + A_{2i}$. Given $A_1,\ldots, A_\ell$, in $O((\ell + \gamma u) \log^7 u\,\mathrm{polyloglog}\,u)$ time, (without actually computing $B_1, \ldots, B_{\ell/2}$) we can
   \begin{enumerate}[label={\normalfont (\roman*)}] 
        \item either tell that $\{B_1, \ldots, B_{\ell/2}\}$ is $(4\gamma)$-sparse, or

        \item return a subset $I$ of $[1,\ell/2]$ such that $|B_i| \geq \frac{2c\gamma u}{|I|} + 1$ for $i \in I$. (Therefore, $\{B_1, \ldots, B_{\ell/2}\}$ is $\gamma$-dense).
     \end{enumerate}    
\end{lemma}
\begin{proof}
    Note that $|B_i| \leq 2u + 1$.  By definition, to determine whether the collection of $B_i$'s is $\gamma$-dense or not,  we should, for each $k \in [2,  2u+1]$, count the number of $B_i$'s with $|B_i| \geq k$.  Due to concern about running time, we check only for those $k$'s that are powers of $2$. More precisely, for $j \in [1, \lfloor \log (2u+1)\rfloor]$, we determine whether $|B_i| \geq 2^j$ or not via Lemma~\ref{lem:est-sumset-size}. We start with $j = 1$. Let $I_j$ be the set of the indices of $B_i$'s with $|B_i| \geq 2^j$.  If the collection $\{B_i\}_{i \in I_j}$ meets the condition of $\gamma$-dense, we are done. Otherwise, we proceed to the next $j$, and obviously, only the $B_i$'s with $i \in I_j$ need to be considered in the next round.  The above procedure is summarized by Algorithm~\ref{alg:est-size}.

    If $(2^{j'}-1)\cdot |I_{j'}| \geq 2c\gamma u$ for some $j'$, since every $B_i$ is a subset of $[0,2u]$ and every $B_i$ with $i \in I_{j'}$ has size at least $2^{j'}$, by definition, $\{B_1, \ldots, B_{\ell/2}\}$ is $\gamma$-dense, and $I_{j'}$ will be returned.

    Suppose that $(2^j-1)\cdot |I_j| < 2c\gamma u$ for all $j$. We will prove that the collection of $B_i$'s is $(4\gamma)$-sparse in this case. Let $\ell_k$ be the number of $B_i$'s whose sizes are at least $k$. It suffices to show that $\ell_k < \frac{8c\gamma u}{k -1}$ for any $k \in [2, u+1]$.  Let $k$ be any integer in $[1, 2u+1]$. Let $j' = \lfloor \log k \rfloor$. We have 
    \[
        \ell_k \leq \ell_{2^{j'}} < \frac{2c\gamma u}{2^{j'} -1} \leq \frac{8c\gamma u}{k -1}.
    \]
    Therefore, $\{B_1, \ldots, B_{\ell/2}\}$ is $(4\gamma)$-sparse.

    Algorithm~\ref{alg:est-size} stops as soon as it identifies that $\{B_1, \ldots, B_{\ell/2}\}$ is $\gamma$-dense, so the outer loop of  has at most $\lfloor\log (2u+1) \rfloor$ iterations. The running time of Algorithm~\ref{alg:est-size} is bounded by the following.
    \begin{align*}
        &\sum_{j = 1}^{\lfloor\log (2u+1) \rfloor} \sum_{i \in I_{j-1}} \left(2^j\log^6 u\,\mathrm{polyloglog}\,u)\right)\\
        =&\left(2\log^6 u\,\mathrm{polyloglog}\,u\right) \left(\sum_{j = 1}^{\lfloor\log (2u+1) \rfloor} (2^{j-1}-1)|I_{j-1}| + \sum_{j = 1}^{\lfloor\log (2u+1) \rfloor}|I_{j-1}|\right)\\
        \leq &2\log^6 u\,\mathrm{polyloglog}\,u \cdot \lfloor\log (2u+1) \rfloor \cdot (2c\gamma u + \ell)\\
        =& O((\ell + \gamma u) \log^7 u\,\mathrm{polyloglog}\,u).\qedhere
    \end{align*}
\end{proof}

\begin{algorithm}
    \caption{\texttt{EstimateDensity}$(A_1, \ldots, A_\ell, \gamma)$}
    \label{alg:est-size}
    \begin{algorithmic}[1]
    \Statex \textbf{Input:} non-empty
    subsets $A_1, \ldots , A_\ell$ of $[0, u]$ and a positive integer $\gamma$.
    \Statex \textbf{Output:} 
    Tells that $\{B_1, \ldots, B_{\ell/2}\}$ is $(4\gamma)$-sparse, or returns a subset $I$ of $[1,\ell/2]$.
    \State $I_{0} := [1, \ell/2]$\;
    \For{$j := 1,...,\lfloor\log (2u+1) \rfloor$}
        \State $I_{j} := \emptyset$\;
        \For{each $i \in I_{j-1}$}
            \If{$|A_{2i-1} + A_{2i}| \geq 2^j$ (via Lemma~\ref{lem:est-sumset-size})}
                \State $I_{j} := I_{j}\cup \{i\}$\;
            \EndIf
        \EndFor
        \If{$|I_{j}|\cdot (2^j-1) \geq 2c\gamma u$}
                \State The algorithm stops and \Return $I_j$\;
        \EndIf
    \EndFor
    \State \Return the collection $\{A_1+A_2, \ldots, A_{\ell-1} + A_{\ell}\}$ is $(4\gamma)$-sparse\;
    \end{algorithmic}
\end{algorithm}

\subsection{Putting Things Together}
Consider level $h$ of the tree for computing $\bigcup_{j=1}^r S^j_1 + \cdots + \bigcup_{j=1}^r S^j_{m}$.  Let $A_1, \ldots, A_\ell$ be nodes in this level where $\ell = \frac{mg}{2^h} \leq mg = O(\beta\log \frac{n}{\eps})$. By Observation~\ref{obs:level-property}, for each $i \in [1, \ell]$, $A_i$ is a subset of $[0, \frac{2^{h+1}}{\eps}]$ with $0 \in A_i$. Moreover, $\sum_{i=1}^{\ell} \max(A_i) \geq \frac{4\beta}{\eps}$.  We make an extra assumption that the elements of $A_i$'s are multiples of $2^{h+1}$.  This can be done by rounding each element of $A_i$ up to the nearest multiple of $2^{h+1}$.  This incurs an additive error of $\ell \cdot 2^{h+1} = 2mg = {O}(\beta\log\frac{n}{\eps})$.  Recall that there are $1 + \log mg$ levels, and each level is repeated for $r$ times. The total additive error caused by the rounding is $2(1 + \log mg)\beta\log\frac{n}{\eps} = O(\beta\log^2\frac{n}{\eps})$.  Let $B_1, \ldots, B_{\ell/2}$ be the next level $h+1$. 

\begin{lemma}\label{lem:tree-operation}
    Let $\gamma = \max(4,\lceil\frac{4mg}{\beta}\rceil)$. Given $A_1, \ldots, A_\ell$, in 
    %$O(\frac{1}{\eps}\log^8\frac{1}{\eps}\,\mathrm{polyloglog}\, \frac{1}{\eps})$
    $\widetilde{O}(\frac{1}{\eps})$ time, we can 
    \begin{enumerate}[label={\normalfont (\roman*)}]
        \item either compute $B_1, \ldots, B_{\ell/2}$ whose total size is $O(\frac{1}{\eps} \log^2 \frac{1}{\eps})$, or

        \item tell that $B_1, \ldots, B_{\ell/2}$ are $\gamma$-dense and that $A_1 + \cdots + A_{\ell}$ has a sequence $z_1 < \cdots < z_k$ such that $z_1 \leq \frac{\beta}{\eps}$, $z_k \geq \frac{2\beta}{\eps}$, and $z_i - z_{i-1} \leq \beta$ for $i \in [2, k]$.
    \end{enumerate}
\end{lemma}
\begin{proof}
    Recall that the elements of $A_i$'s are multiples of $2^{h+1}$.  For $i \in [1, \ell]$, define $A'_i = \{\frac{a }{2^{h+1}} : a \in A_i\}$. $A'_1, \ldots, A'_\ell$ are subsets of $[0, \frac{1}{\eps}]$.  For $i \in [1, \ell/2]$, define $B'_{i} = A'_{2i-1} + A'_{2i}$. We estimate the density of $\{B'_1, \ldots, B'_{\ell/2}\}$ via Lemma~\ref{lem:est-density}. The time cost is
    \[
        O((\ell + \frac{\gamma}{\eps}) \log^7 \frac{1}{\eps}\,\mathrm{polyloglog}\, \frac{1}{\eps}) = \widetilde{O}(\frac{1}{\eps}).
    \]
    
    If $\{B'_1, \ldots, B'_{\ell/2}\}$ is $(4\gamma)$-sparse, by Lemma~\ref{lem:sparse-small-size},  
    \[
        \sum_{i=1}^{\ell/2}|B'_i| \leq \frac{\ell}{2} +  4c\gamma \cdot \frac{2}{\eps} \cdot  (1 + \log \frac{2}{\eps}) = O(\frac{1}{\eps} \log^2 \frac{1}{\eps}).
    \]
    We can compute each $B'_i$ via Lemma~\ref{lem:sparse-fft}. The time cost for computing all $B'_i$'s is 
    \[
        O(\sum_{i=1}^{\ell/2} |B'_i| \log^5 \frac{1}{\eps}\,\mathrm{polyloglog}\, \frac{1}{\eps}) = \widetilde{O}(\frac{1}{\eps}).
    \] 
    From $B'_i$, we can easily obtain $B_i$ by multiply each element of $B'_i$ by $2^{h+1}$.

    Consider the case that $\{B'_1, \ldots, B'_{\ell/2}\}$ is $\gamma$-dense.  Let $\eta = \max(B'_1) + \cdots + \max(B'_{\ell/2})$.  We have that 
    \begin{equation}\label{eq:bound-sum-max}
            \frac{\beta}{2^{h-1} \eps } \leq \eta \leq \frac{\ell}{\eps}.
    \end{equation}
    The first inequality is due to that $\max(A_1) + \cdots + \max(A_\ell) \geq \frac{4\beta}{\eps}$, and the second is due to that $\max(A'_i) \leq \frac{1}{\eps}$.
    By Lemma~\ref{lem:dense2longseq}, $B'_1 + \cdots + B'_{\ell/2}$ contains a sequence $z'_1, \ldots, z'_k$ such that $z'_1 \leq \frac{2\eta}{\gamma}$, $z'_k \geq \frac{(\gamma-2)\eta}{\gamma}$, and that $z'_{i} - z'_{i-1} \leq \frac{2\ell}{\gamma}$ for $i \in [2,k]$.  For $i \in [1, k]$, let $z_i =2^{h+1} \cdot z'_i$.  One can see that $z_1, \ldots, z_k$ is a sequence contained in $B_1 + \cdots + B_{\ell/2} = A_1 + \cdots + A_{\ell}$. By inequality~\eqref{eq:bound-sum-max} and our choice of $\gamma$, we have the followings.
    \begin{align*}
        z_i - z_{i-1} &\leq \frac{2\ell}{\gamma}\cdot 2^{h+1} = \frac{4mg}{\gamma} \leq \beta\\
        z_1 &\leq  \frac{2\eta}{\gamma} \cdot 2^{h+1} \leq \frac{4mg}{\gamma\eps} \leq \frac{\beta}{\eps}\\
        z_k &\geq  \frac{(\gamma-2)\eta}{\gamma} \cdot 2^{h+1} \geq \frac{\gamma-2}{\gamma}\cdot \frac{4\beta}{\eps} \geq \frac{2\beta}{\eps}.\qedhere
    \end{align*}
\end{proof}

\begin{lemma}\label{lem:main}
    Given the $r$ partitions of $X$ resulting from color-coding,  we can obtain a set of size $O(\frac{1}{\eps}\log^2 \frac{1}{\eps})$ that approximates $\S_X[\frac{\beta}{\eps},\frac{2\beta}{\eps}]$ with additive error $O(\beta\log^2\frac{n}{\eps})$. The time cost is $\widetilde{O}(n + \frac{1}{\eps})$.
\end{lemma}
\begin{proof}
    Each partition form a level $0$ of the tree for computing $\bigcup_{j=1}^r S^j_1 + \cdots + \bigcup_{j=1}^r S^j_{m}$.  We keep computing new levels until we reach the root or we encounter a level that is dense (That is, case (ii) of Lemma~\ref{lem:tree-operation}). Recall that there are $1 + \log mg$ levels and that the first $1 + \log g$ levels may be repeated, but only for at most $r$ times. Every level we have computed has a total size of $O(\frac{1}{\eps} \log^2 \frac{1}{\eps})$ (except for level $0$, which has a total size of $n$).  The total time cost for rounding the levels is
    \[
        O(r\cdot (n +  \frac{1}{\eps} \log^2 \frac{1}{\eps}\log mg)) = \widetilde{O}(n + \frac{1}{\eps}),
    \]
    and that for computing the levels is at most 
    \[
        r\log mg \cdot \widetilde{O}(\frac{1}{\eps}) = \widetilde{O}(\frac{1}{\eps}).
    \]

    Recall that rounding the levels incurs a total additive error of $O(\beta\log^2\frac{n}{\eps})$. If we reach the root of the tree, then we get a set that approximates $\S_X[\frac{\beta}{\eps}, \frac{2\beta}{\eps}]$ with additive error $O(\beta\log^2\frac{n}{\eps})$.  The size of this set is $O(\frac{1}{\eps}\log^2 \frac{1}{\eps})$. Consider the other case that we encounter at a dense level. Lemma~\ref{lem:tree-operation} implies that $A_1 + \cdots + A_{\ell}$ has a sequence $z_1 < \cdots < z_k$ such that $z_1 \leq \frac{\beta}{\eps}$, $z_k \geq \frac{2\beta}{\eps}$, and $z_i - z_{i-1} \leq \beta$ for $i \in [2, k]$.  We shall approximate $\S_X[\frac{\beta}{\eps}, \frac{2\beta}{\eps}]$ in the same spirit as Lemma~\ref{lem:easy-approx}. Let $S_A = A_1 + \cdots + A_{\ell}$.  By Lemma~\ref{lem:easy-approx}, the following set of size $O(\frac{1}{\eps})$ approximates $S_A[\frac{\beta}{\eps}, \frac{2\beta}{\eps}]$ with additive error $2\beta$.
    \[
        \widetilde{S} = \{\frac{\beta}{\eps} + \beta i : i \in [0, \frac{2}{\eps}]\}.
    \]
    We shall show that $\widetilde{S}$ approximates $\S_X[\frac{\beta}{\eps}, \frac{2\beta}{\eps}]$ with additive error $O(\beta\log^2 \frac{n}{\eps})$. Clearly, for any $s \in \S_X[\frac{\beta}{\eps}, \frac{2\beta}{\eps}]$, there is $\tilde{s} \in \widetilde{S}$ such that 
    \[
        s- O(\beta\log^2 \frac{n}{\eps}) \leq s - \beta \leq \tilde{s} \leq s + \beta \leq s + O(\beta\log^2 \frac{n}{\eps}).
    \]
    So the first condition of Definition~\ref{def:approx} is satisfied. Now consider any $\tilde{s} \in \widetilde{S}$.  Since $\widetilde{S}$ approximates $S_A[\frac{\beta}{\eps}, \frac{2\beta}{\eps}]$ with additive error $\beta$, there exist $s_A \in S_A[\frac{\beta}{\eps}, \frac{2\beta}{\eps}]$ such that 
    \[
           \tilde{s} - \beta \leq s_A \leq \tilde{s} + \beta.
    \]
    The error between $S_A$ and $\S_X$ is due to rounding, which is bounded by $O(\beta\log^2 \frac{n}{\eps})$. That is to say for every $s_A \in S_A$, there is $s \in \S_X$ such that 
    \[
            s_A- O(\beta\log^2 \frac{n}{\eps}) \leq s \leq s_A + O(\beta\log^2 \frac{n}{\eps}).
    \]
    So there is $s \in \S_X$ such that 
    \[
           \tilde{s} - \beta - O(\beta\log^2 \frac{n}{\eps}) \leq s \leq \tilde{s} + \beta + O(\beta\log^2 \frac{n}{\eps}).
    \]
    The second condition of Definition~\ref{def:approx} is also satisfied.
\end{proof}

\section{Recovering a Solution}\label{sec:algo-solution}
Let $\widetilde{S}$ be the set we obtained via Lemma~\ref{lem:main}.  Let $\delta^* = O(\beta\log^2\frac{n}{\eps})$ be the additive error.  This section gives a approach for recovering a subset $Y$ of $X$ with $\tilde{s} - \delta^* \leq \Sigma(Y) \leq \tilde{s} + \delta^*$ for any $\tilde{s} \in \widetilde{S}$.  

\subsection{A High-Level Overview}
\begin{observation}
    Given any $y \in A + B$, one can recover $a \in A$ and $b \in B$ such that $y = a + b$ as follows: sort $A$ and $B$, for each $a\in A$, check if $y - a \in B$ using binary search. The running time is $O(|A| \log |A| + |B|\log |B|)$, which is no larger than the time needed to compute $A + B$.  
\end{observation}

If $\widetilde{S}$ is obtained by explicitly computing all levels, then by the above observation, $Y$ can be recovered by tracing back the computation of levels in $\widetilde{O}(n + \frac{1}{\eps})$ time. This method, however, does not work when some level is dense. This is because in the dense case, we use additive combinatorial tools to show the existence of a particular sequence (Lemma~\ref{lem:tree-operation}(ii)), skip the remaining levels, and directly give the approximation set $\widetilde{S}$.  

In order to make ``tracing back'' work, we tackle dense levels with the following alternative approach: even when a level is dense, it will be computed but only partially. We use additive combinatorics tools to ensure that the partial level has a small total size but is still dense enough to provide a good approximation.

\subsection{Computing Partial Levels}
We first show that given an integer $k$, it is possible to compute a partial sumset whose size is $k$.

\begin{lemma}\label{lem:subset-of-sumset}
    Let $A$ and $B$ be two non-empty subsets of $[0,u]$ containing $0$.  Let $k \leq |A + B|$ be a positive integer. In $O((k + |A| + |B|)\log^7 u\,\mathrm{polyloglog}\,u)$ time, we can compute a set $H \subseteq A+B$ such that $|H| = k$ and that $H$ contains $0$ and $\max(A) + \max(B)$.
\end{lemma}
\begin{proof}
    If $|A| \geq k$, $H$ can be easily obtained from $A \cup \{\max(A) + \max(B)\}$. Assume that $|A| < k$. Let $m = |B|$. For $i \in [1, m]$, define $B_i$ be the set of the first $i$ elements of $B$.  Note that $|A + B_i|$ is an increasing function of $i$. Moreover, let $i^*$ be the smallest $i$ with $|A + B_i| \geq k$.  It must be that $k \leq |A + B_{i^*}| \leq k + |A|$.  And $i^*$ must exist since $k \leq |A + B|$.  We use binary search to find $i^*$.  Start with $i = m/2$. We can determine whether $|A + B_i| \geq k$ via Lemma~\ref{lem:est-sumset-size}.  We search the left half if $|A + B_i| \geq k$, and search the right half otherwise. The time cost for search $i^*$ is 
    \[
        \log m \cdot O(k\log^6 u\,\mathrm{polyloglog}\,u) = O(k\log^7 u\,\mathrm{polyloglog}\,u).
    \] 
    Then we compute $A + B_{i^*}$ via Lemma~\ref{lem:sparse-fft}. Given $A + B_{i^*}$, $H$ will be easy to constructed. The time cost for computing $|A + B_{i^*}|$ is 
    \[
            O((k + |A|)\log^5 u\,\mathrm{polyloglog}\,u) =  O((k + |A| + |B|)\log^5 u\,\mathrm{polyloglog}\,u).\qedhere
    \]
\end{proof}

Next we show that given a level that nicely approximates $\S_X[\frac{\beta}{\eps}, \frac{2\beta}{\eps}]$, we can always compute a new (perhaps partial) level of small size that also nicely approximates $\S_X[\frac{\beta}{\eps}, \frac{2\beta}{\eps}]$, even when the next level is dense.  Recall that in level $h$,  there are $\ell$ nodes $A_1, \ldots, A_\ell$ where $\ell = \frac{mg}{2^h} \leq mg = O(\beta\log \frac{1}{\eps})$. For each $i \in [1, \ell]$, $A_i$ is a subset of $[0, \frac{2^{h+1}}{\eps}]$ with $0 \in A_i$. Moreover, $\sum_{i=1}^{\ell} \max(A_i) \geq \frac{4\beta}{\eps}$ and the elements of $A_i$'s are assumed to be multiples of $2^{h+1}$.

\begin{lemma}\label{lem:compute-dense-level}
    Let $A_1, \ldots, A_\ell$ be a level such that $A_1 + \cdots + A_{\ell}$ approximates $\S_X[\frac{\beta}{\eps}, \frac{2\beta}{\eps}]$ with additive error $\delta$. 
    In $\widetilde{O}(\frac{1}{\eps} + \sum_{i=1}^{\ell} |A_\ell|)$ time, we can compute $Z_1, \ldots, Z_{\ell/2}$ such that the following is true.
    \begin{enumerate}[label={\normalfont (\roman*)}]
        \item For $i\in[1,\ell/2]$, $Z_i \subseteq A_{2i-1} + A_{2i}$ and $\{0, \max(A_{2i-1}) + \max(A_{2i})\}\subseteq Z_i$.

        \item $\sum_{i=1}^{\ell/2} |Z_{i}| = O(\frac{1}{\eps} \log^2\frac{1}{\eps}) $

        \item $Z_1 + \cdots + Z_{\ell/2}$ approximates $\S_X[\frac{\beta}{\eps}, \frac{2\beta}{\eps}]$ with additive error $\max(\delta, \beta)$.
    \end{enumerate} 
\end{lemma}
\begin{proof}
    The proof is almost the same as that of Lemma~\ref{lem:tree-operation}. The only difference is that, in the dense case, we will compute a small subset of $A_{2i-1} + A_{2i}$ so that the collection of these subsets has a small total size but is still dense.

    Recall that the elements of $A_i$'s are multiples of $2^{h+1}$.  For $i \in [1, \ell]$, define $A'_i = \{\frac{a }{2^{h+1}} : a \in A_i\}$. $A'_1, \ldots, A'_\ell$ are subsets of $[0, \frac{1}{\eps}]$. Let $\gamma = \max(4,\lceil\frac{4mg}{\beta}\rceil)$.  For $i \in [1, \ell/2]$, define $B'_{i} = A'_{2i-1} + A'_{2i}$. We estimate the density of $\{B'_1, \ldots, B'_{\ell/2}\}$ via Lemma~\ref{lem:est-density}. If $\{B'_1, \ldots, B'_{\ell/2}\}$ is $\gamma$-sparse, we compute $B'_1, \ldots, B'_{\ell/2}$.  Let $Z_i = \{2^{h+1} \cdot b : b \in B'_i\}$. One can see that the $Z_i$'s satisfy properties (i)(ii)(iii). And the total time cost is $\widetilde{O}(\frac{1}{\eps})$.

    Suppose that $\{B'_1, \ldots, B'_{\ell/2}\}$ is $\gamma$-dense. Let $u = \frac{2}{\eps}$. In this case, Lemma~\ref{lem:est-density} also returns a subset $I$ of $[1,\ell/2]$ such that $|B'_i| \geq \frac{2c\gamma u}{|I|} + 1$ for $i \in I$. We compute $Z'_i$ for all $i \in [1,\ell/2]$ as follows.  For $i \in I$, we compute a $Z'_i \subseteq A'_{2i-1} + A'_{2i}$ via Lemma~\ref{lem:subset-of-sumset} such that 
    \[
        |Z'_i| = \frac{2c\gamma u}{|I|} + 1.
    \] 
    For $i\notin I$, we let $Z'_i = \{0, \max(A'_{2i-1}) + \max(A'_{2i})\}$. Finally we returns $Z_i = \{2^{h+1} \cdot z : z \in Z'_i\}$ for all $i$. The total running time is
    \begin{align*}
        &\frac{\ell}{2} - |I| + \sum_{i \in I} ( \frac{2c\gamma u}{|I|} + 1 + |A'_{2i-1}| + |A'_{2i}|)\log^7 u\,\mathrm{polyloglog}\,u\\
    \leq & \ell + (2c\gamma u + |I| + \sum_{i=1}^{\ell} |A'_i|)\log^7 u\,\mathrm{polyloglog}\,u\\
    \leq & \widetilde{O}(\frac{1}{\eps}  + \sum_{i=1}^{\ell} |A_\ell|).
    \end{align*}
    Next, we show that the $Z_i$'s satisfy properties (i)(ii)(iii).  Property (i) is guaranteed by Lemma~\ref{lem:subset-of-sumset}. Property (ii) is satisfied since
    \begin{align*}
        \sum_{i=1}^{\ell/2}|Z_i| = \sum_{i=1}^{\ell/2}|Z'_i| \leq \frac{\ell}{2} - |I| + \sum_{i \in I} ( \frac{2c\gamma u}{|I|} + 1) = O(\frac{1}{\eps} \log\frac{1}{\eps}).
    \end{align*}
    We are left to prove property (iii).  Since $|Z'_i| \geq \frac{2c\gamma u}{|I|} + 1 $ for $i \in I$, the collection $\{Z'_1, \ldots, Z'_{\ell/2}\}$ is $\gamma$-dense. By the construction of $Z'_i$, we have $0 \in Z'_i$ for any $i \in [1,\ell/2]$ and $\sum_{i=1}^{\ell/2}\max(Z'_i) = \sum_{i=1}^{\ell}\max(A'_i)$.  Therefore, $Z'_1, \ldots, Z'_{\ell/2}$ has exactly the property as the $\gamma$-dense collection $\{B'_1, \ldots, B'_{\ell/2}\}$ in the proof of Lemma~\ref{lem:tree-operation}.  That is, $Z_1 + \cdots + Z_{\ell/2}$  
    has a sequence $z_1 < \ldots < z_k$ such that $z_1 \leq \frac{\beta}{\eps}$, $z_k \geq \frac{2\beta}{\eps}$, and $z_i - z_{i-1} \leq \beta$. Clearly, for any $s\in \S_X[\frac{\beta}{\eps}, \frac{2\beta}{\eps}]$, there is some $z \in Z_1 + \cdots + Z_{\ell/2}$ such that 
    \[
       s- \beta  \leq  z \leq s + \beta.
    \]
    For any $z \in Z_1 + \cdots + Z_{\ell/2}$, $z$ is also in $A_1 + \cdots + A_\ell$ which approximates $\S_X[\frac{\beta}{\eps}, \frac{2\beta}{\eps}]$ with additive error $\delta$. So there is some $s \in \S_X$ such that
    \[
        z-\delta \leq  s \leq z + \delta.      
    \]
    By Definition~\ref{def:approx}, $ Z_1 + \cdots + Z_{\ell/2}$ approximates $\S_X[\frac{\beta}{\eps}, \frac{2\beta}{\eps}]$ with additive error $\max(\delta, \beta)$.
\end{proof}

\subsection{Tackling Dense Levels via Partial Levels}

Given color-coding partitions of $X$ as level $0$, we compute a set $\widetilde{S}$ in a tree-like manner as in Section~\ref{sec:alg-determine}. Now no matter a level is dense or sparse, we can always compute it (partially) via Lemma~\ref{lem:compute-dense-level}.
%we shall compute new levels via Lemma~\ref{lem:tree-operation} until we encounter some dense level. That is, Lemma~\ref{lem:tree-operation}(ii) occurs, claiming that $B_1, \ldots, B_{\ell/2}$ are $\gamma$-dense and that $A_1 + \cdots + A_{\ell}$ has a sequence $z_1 < \cdots < z_k$ such that $z_1 \leq \frac{\beta}{\eps}$, $z_k \geq \frac{2\beta}{\eps}$, and $z_i - z_{i-1} \leq 2\beta$ for $i \in [2, k]$. One can verify that $A_1 + \cdots + A_{\ell}$ approximates $\S_X[\frac{\beta}{\eps}, \frac{2\beta}{\eps}]$ with 
The total additive error is $\max(\beta, O(\beta\log^2 \frac{n}{\eps})) = O(\beta\log^2 \frac{n}{\eps}) $, which results from Lemma~\ref{lem:compute-dense-level} and the rounding of the integers in each level. %Then we compute (partially) all levels above $A_1, \ldots, A_\ell$ via Lemma~\ref{lem:compute-dense-level}. In total, we invoke Lemma~\ref{lem:compute-dense-level} $\log \ell$ times.  Finally, we can obtains a set $\widetilde{S}$ that approximates $\S_X[\frac{\beta}{\eps}, \frac{2\beta}{\eps}]$ with additive error $O(\beta\log^3 \frac{1}{\eps})$. 
Since all the levels below $\widetilde{S}$ are explicitly computed, we can use ``tracing back'' to recover the corresponding subset $Y$ for each $\tilde{s} \in \widetilde{S}$.

Now we analyze the total running time of the above procedure. Lemma~\ref{lem:compute-dense-level} can be invoked for at most $\log \ell$ times.  Each level has a size of $O(\frac{1}{\eps}\log^2 \frac{1}{\eps})$ (except for $A_1, \ldots, A_\ell$, which may have a total size of $O(n)$). The total running time is bound by 
\[
    \widetilde{O}((\frac{1}{\eps} + \frac{1}{\eps} \log^2\frac{1}{\eps}) \cdot\log \ell + \frac{1}{\eps} + n ) \\
    = \widetilde{O}(n + \frac{1}{\eps}).
\]

\section{Conclusion}\label{sec:conclude}
We present a near-linear time weak approximation scheme for Subset Sum in this paper. It is interesting to explore whether our technique can bring improvement to other related problems. In particular, it is a major open problem whether there exists an $O(n+x_{\max})$-time exact algorithm for Subset Sum, where $x_{\max}$ refers to the largest input integer. The best known algorithm has a running time $O(n+x_{\max}^{3/2})$~\cite{CLMZ23}. It would be interesting to close or narrow the gap.   

\clearpage
\appendix

\section{Problem Reduction}\label{sec:reduce}
\lemreduce*
\begin{proof}
%As mentioned above, we approach the Subset Sum problem in two phases. In the first phase, we compute a set $\widetilde{S}$ that approximates $\S_X[0,t]$ with additive error $\eps t$, and also find $\tilde{s}$, which is the maximum element of $\widetilde{S}[0,(1+\eps)t]$. In the second phase, we recover subset $Y$ from $\tilde{s}$, such that
%\[
%\tilde{s}-\eps t \leq \Sigma(Y)\leq \tilde{s}+\eps t.
%\]
In this proof, when we say we can $\eps$-solve a Subset Sum instance $(X, t)$ , we mean that we can compute a set of size $\widetilde{O}(\frac{1}{\eps})$ that approximates $\S_X[0,t]$ with additive error $\widetilde{O}(\eps)\cdot t$, and given any element of the computed set, we can recover a solution with additive error $\widetilde{O}(\eps)\cdot t$. 
%It is worth noticing that since we are computing a set that approximates $\S_X$, the target $t$ does not matter too much -- we can afford to change it by $O(1)$ times at the cost of increasing the additive error by $O(1)$ times.  

Let $(X, t)$ be the original instance. Let $Z = \{x < \eps t : x\in X\}$ be the set of tiny integers in $X$.  If $\Sigma(Z) \geq \eps t$, we iteratively extract from $Z$ a minimal subset $Z'$ such that $\Sigma(Z') \geq \eps t$ until $\Sigma({Z})< \eps t$.  Now it is safe to discard ${Z}$ as it incurs only an additive error of $\eps t$. Let $Z_1,Z_2,\cdots,Z_h$ be the sets we extracted from $Z$. For each $Z_j$, we replace it with a meta-integer of value $\Sigma(Z_j)$. Note that $\Sigma(Z_j) < 2\eps t$ for every $j$ as $Z_j$ is minimal.  It is easy to see that the above transformation can be done in $O(n)$ time. Moreover, replacing the tiny integers by the meta-integers incurs an additive error of at most $2\eps t$, because for any subset $Z'$ of tiny integers, there is a maximal set of meta-integers that sums to at least $\Sigma(Z) - 2\eps t$.

%One can easily verify that the total additive error caused by this step is at most $2\eps t$ and it takes time in $O(n)$.  

Now $x \geq \eps t$ for all $x \in X$, which implies that $\frac{x}{\eps^2t} \geq \frac{1}{\eps}$. By adjusting $\eps$ by a constant factor, we assume that $\frac{1}{\eps}$ is an integer.  %By adjusting $\eps$ by an $O(1)$ factor, we assume that $\frac{1}{\eps}$ is an integer. 
Now we scale the whole instance by $\eps^2 t$, that is, we replace $x\in X$ by $x':= \lfloor \frac{x}{\eps^2 t}\rfloor$ and $t$ by $t' = \lfloor \frac{t}{\eps^2t}\rfloor$. Scaling incurs an approximate factor of at most $1 + \eps$, or equivalently, an additive error of at most $\eps t$. After scaling, we have $x \in [\frac{1}{\eps}, \frac{1}{\eps^2}]$ for all $x \in X$ and $t = \frac{1}{\eps^2}$.   

Then we divide $X$ into $\log \frac{1}{\eps}$ groups, where each group contains integers within $[\frac{\alpha}{\eps}, \frac{2\alpha}{\eps}]$ for $\alpha \in \{1,2,4,8,\cdots\} \cap [1, \frac{1}{\eps}]$. We denote $X \cap [\frac{\alpha}{\eps}, \frac{2\alpha}{\eps}]$ as $X_{\alpha}$. Suppose we can $\eps$-solve each $(X_{\alpha}, \frac{1}{\eps^2})$ in $\widetilde{O}(n+\frac{1}{\eps})$-time (and let $\tilde{S}_{\alpha}$ denote the set that approximates $\S_{X_{\alpha}}$), then we can also $\eps$-solve $(X, \frac{1}{\eps^2})$ in $\widetilde{O}(n+\frac{1}{\eps})$-time via Lemma~\ref{lem:approx-k-fft} (by letting $A_j$'s be $\tilde{S}_{\alpha}$'s).

From now on we focus explicitly on $X$ where each $x\in X$ satisfies that $x \in [\frac{\alpha}{\eps}, \frac{2\alpha}{\eps}]$ for some $\alpha \in [1, \frac{1}{\eps}]$. %This can be done by dividing $X$ into $\log \frac{1}{\eps}$ groups, and merge their results via Lemma~\ref{lem:approx-k-fft}.  We can show that the size of each result is $\widetilde{O}(\frac{1}{\eps})$, so the merging time is in $\widetilde{O}(\frac{1}{\eps})$.
For every $x\in X$, we further scale the instance by $\alpha$. That is, we let $x'=\lfloor x/\alpha\rfloor$ for every $x \in X$  and $t' = \lfloor\frac{1}{\alpha\eps^2}\rfloor$. Again, scaling incurs an approximate factor of at most $1 + \eps$, or equivalently, an additive error of at most $\eps t$. 

Now $x\in [\frac{1}{\eps}, \frac{2}{\eps}]$ for every $x\in X$ and $t=\lfloor\frac{1}{\alpha\eps^2}\rfloor\in [\frac{1}{\eps},\frac{1}{\eps^2}]$. We further assume that $t\leq  \Sigma(X)/2$. If $t>\Sigma(X)/2$, then to approximate $\S_X[0,t]$, it suffices to approximate $\S_X[0,\Sigma(X)/2]$ and $\S_X[\Sigma(X)/2, t]$ separately. The following claim indicates that it actually suffices to approximate $\S_X[0,\Sigma(X)/2]$.

%Recall that by $\eps$-solving $X$ with respect to $t$ we mean to compute a set $\widetilde{S}$ that approximates $\S_{X}$ with additive error $\eps t$, and recover a solution with additive error $\eps t$ for any $\tilde{s}\in \widetilde{S}$.

%Assume without loss of generality that $t\leq \Sigma(X)$. 
%To compute a set that approximates $\S_{X}$ with additive error We show that if $t> \Sigma(X)/2$, we can just approximate $\S_X[0,\Sigma(X)/2]$ by the following claim. 
\begin{claim}\label{claim:half}
Let  $\widetilde{S}$ be a set that approximates $\S_X[0,\Sigma(X)/2]$ with additive error $\eps t$. Then $\{\Sigma(X)-s':s'\in \widetilde{S}\}$ approximates $\S_X[\Sigma(X)/2, \Sigma(X)]$ with additive error $\eps t$.
\end{claim}
\begin{proof}[Proof of Claim~\ref{claim:half}]
    We prove Claim~\ref{claim:half} by Definition~\ref{def:approx}.
Consider an arbitrary $s\in \S_X[\Sigma(X)/2, \Sigma(X)]$, we have $\Sigma(X)-s\in \S_X[0, \Sigma(X)/2]$. Hence, there is $\tilde{s}\in \widetilde{S}$ with $\Sigma(X)-s-\eps t\leq \tilde{s}\leq \Sigma(X)-s+\eps t$. Rearranging the inequality we get $s-\eps t \leq \Sigma(X)-\tilde{s}\leq s+\eps t$. %For any $\tilde{s}\in \{\Sigma(X)-s':s'\in \widetilde{S}\}$, $\Sigma(X)-\tilde{s} \in \widetilde{S}$. There is an $s \in \S_X$ with $\Sigma(X)-\tilde{s}-\eps t\leq s \leq \Sigma(X)-\tilde{s}+\eps t$. We can get $\tilde{s}-\eps t\leq \Sigma(X)-s \leq \tilde{s}+\eps t$. Obviously, $\Sigma(X)-s\in \S_X$. 
Hence, $\{\Sigma(X)-s':s'\in \widetilde{S}\}$ approximates $\S_X[\Sigma(X)/2, \Sigma(X)]$ with additive error $\eps t$.
\end{proof}

From now on we focus explicitly on approximating $\S_X[0,t]$ for $t\leq \Sigma(X)/2$.
%Let  $\widetilde{S}$ be a set that approximates $\S_X[0,\Sigma(X)/2]$ with additive error $\eps t$. Then $(\widetilde{S}\cup\{\Sigma(X)-s':s'\in \widetilde{S}\})[0,(1+\eps)t]$ approximates $\S_X[0, t]$ with additive error $\eps t$. So we can assume $t\leq \Sigma(X)/2$.
To approximate $\S_X[0,t]$, we partition $\S_X[0,t]$ into $\log t \leq 2\log \frac{1}{\eps}$ subsets $\S_X[0,\frac{1}{\eps}], \S_X[\frac{1}{\eps},\frac{2}{\eps}],\S_X[\frac{2}{\eps},\frac{4}{\eps}],\ldots$. Note that $\S_X[0,\frac{1}{\eps}]$ can be computed directly as $x\in [\frac{1}{\eps}, \frac{2}{\eps}]$ for every $x\in X$, so we focus on the remaining subsets. Each of the remaining subsets can be denoted as $\S_X[\frac{\beta}{\eps}, \frac{2\beta}{\eps}]$ for some integer $\beta\in [1, \frac{1}{\eps}]$. To approximate $\S_X[0,t]$ with an additive error of $\widetilde{O}(\eps t)$, it suffices to approximate each $\S_X[\frac{\beta}{\eps}, \frac{2\beta}{\eps}]$ with an additive error of $\widetilde{O}(\beta)$. Notice that since $\frac{2\beta}{\eps}\leq t$ and $t\leq \Sigma(X)/2$, it follows directly that $\Sigma(X)\geq \frac{4\beta}{\eps}$ for every $\S_X[\frac{\beta}{\eps}, \frac{2\beta}{\eps}]$.
%Furthermore, to $\eps$-solve $X$ within each $[\frac{\beta}{\eps}, \frac{2\beta}{\eps}]$, it suffices to approximate $\S_X[\frac{\beta}{\eps}, \frac{2\beta}{\eps}]$ and to recover a solution with additive error $\delta=\widetilde{O}(\beta)\leq \widetilde{O}(\eps)\cdot t$.
%For an integer $\beta\in [1, \eps t]$, it suffices to approximate some $\S_X[\frac{\beta}{\eps}, \frac{2\beta}{\eps}]$ with additive error $\widetilde{O}(\beta)$. As the size of each subsets is $\widetilde{O}(\frac{1}{\eps})$, the union time is in $\widetilde{O}(\frac{1}{\eps})$. As $t\leq \Sigma(X)/2$, we have $\Sigma(X)\geq \frac{4\beta}{\eps}$.
%Let $\tilde{s}$ be the maximum element of $\widetilde{S}[0,(1+\eps)t]$, now we prove that we can recover subset $Y$ from $\tilde{s}$ in time $\widetilde{O}(n+\frac{1}{\eps})$. We compute $\widetilde{S}$ by Lemma~\ref{lem:approx-k-fft} in a binary-tree-like manner. So by binary search, we can split $\tilde{s}$ into each group in $\widetilde{O}(\frac{1}{\eps})$. Then in each group, we can recover a subset $Y'\subseteq X'$ in $\widetilde{O}(n+\frac{1}{\eps})$. The union of these subsets is the subset $Y$.
\end{proof}

\section{Details for Color-Coding}\label{app:color-coding}
Let $\beta\in [1,\frac{1}{\eps}]$. Fix $q^*: = (\frac{1}{\eps})^{-O(1)}$.  Let $m := 4\beta/\log \frac{4\beta^2}{\eps q^*}$ rounded up to be next power of $2$, $g:= 36\log^2\frac{4\beta^2}{\eps q^*}$ rounded up to the next power of $2$, and $r := \lceil \log \frac{4\beta^2}{\eps q^*} \rceil$. 
 The goal of this section is to prove the following.
\lemcolorcoding*

We are essentially following the color coding method in the prior work~\cite{Bri17} except that we need to additionally enforce the property that 
%We need an additional property that for any $j\in[1,r]$, 
\[
\max(X^j_{1,1}) + \cdots + \max(X^j_{1,g}) + \cdots +\max(X^j_{m,1}) + \cdots + \max(X^j_{m, g}) \geq \frac{4\beta}{\eps}.
\]
which can be achieved by ensuring that there are enough non-empty subsets. This is not difficult to achieve, in particular, we may simply do the following: for half of the subsets, we let each of them contain exactly one integer; for the remaining half of the subsets, we use color coding. The above procedure is summarized in Algorithm~\ref{alg:mod-color-coding}.  

%Towards the proof, we identify two cases. If $|X|\leq mg/2$, then we simply run the following Algorithm~\ref{alg:small-color-coding} that partitions $X$ into $mg$ subsets where every subset contains zero or one element, and in this case Lemma~\ref{lem:mod-color-coding} is straightforwards. We summarize the above procedure as See Algorithm~\ref{alg:mod-color-coding}.  %When $|X| \leq mg$, we use Algorithm~\ref{alg:small-color-coding} and every partition has the stated property since $\Sigma(X) \geq \frac{4\beta}{\eps}$ (See the problem in Lemma~\ref{lem:reduce}).

\begin{algorithm}[H]
\caption{$\mathtt{SmallColorCoding}(X, m, g, r)$}
\label{alg:small-color-coding}
    \begin{algorithmic}[1]
    \Statex \textbf{Input:} A multi-set $X$ of integers that $|X|\leq mg$ and three integers $m, g, r$
    \Statex \textbf{Output:} $r$ partitions $\{X^j_{1,1}, \ldots, X^j_{1,g}, \ldots, X^j_{m,1}, \ldots, X^j_{m, g}\}_{j \in [1, r]}$ of $X$
    \State Partition $X$ into $X_{1,1}, \ldots, X_{m,g}$ that each subset $|X_{i,j}|\leq 1$\;
        \For{$j =1,\ldots,r$}
        \State $X^j_{1,1}, \ldots, X^j_{m, g} := X_{1,1}, \ldots, X_{m, g}$
        \EndFor
  \State \Return the $r$ partitions of $X$. 
  \end{algorithmic}
\end{algorithm}

\begin{algorithm}
\caption{$\mathtt{ModifiedColorCoding}(X, k, q)$}
\label{alg:mod-color-coding}
\begin{algorithmic}[1]
    \Statex \textbf{Input:} A multi-set $X$ of integers, a positive integer $k$, and a target error probability $q$
    \Statex \textbf{Output:} $r$ partitions $\{X^j_{1,1}, \ldots, X^j_{2m, g}\}_{j \in [1, r]}$ of $X$ satisfying Lemma~\ref{lem:mod-color-coding}
    \State $m:= 2k/\log(k/q)$ rounded up to be next power of $2$
    \State $g:= 36\log^2(k/q)$ rounded up to the next power of $2$
    \State $r:= \lceil \log \frac{k}{q} \rceil$
    \If{$|X|\leq mg/2$}
        \State $\{X^j_{1,1}, \ldots,  X^j_{m, g}\}_{j \in [1, r]} := \mathtt{SmallColorCoding}(X, m, g, r)$\;
    \Else
        \State Let $X_1$ be an arbitrary subset of $X$ with $|X_1| = mg/2$\;
        \State $X_2 := X \setminus X_1$\;
        \State $\{X^j_{1,1}, \ldots,  X^j_{m/2, g}\}_{j \in [1, r]} := \mathtt{SmallColorCoding}(X, m/2, g, r)$\;
        \State $\{X^j_{m/2+1,1}, \ldots,  X^j_{m, g}\}_{j \in [1, r]} := \mathtt{ColorCoding}(X_2, k,q)$\;
    \EndIf
    \State \Return $\{X^j_{1,1}, \ldots, X^j_{m, g}\}_{j \in [1, r]}$\;
\end{algorithmic}
\end{algorithm}

%Recall that if $Y$ is a subset of $X$ with $\Sigma(Y) \in \S_X[\frac{\beta}{\eps}, \frac{2\beta}{\eps}]$, $|Y| \leq 2\beta$. So we let $k:=2\beta$. By adjusting the error probability $q^*$ by a factor of $\frac{2\beta}{\eps}\geq \frac{\beta}{\eps} + 1$, with high probability of $1-q^*$, we can obtain not only a particular $\Sigma(Y)$, but any $\Sigma(Y)\in \S_X[\frac{\beta}{\eps}, \frac{2\beta}{\eps}]$, that is, the set $\S_X[\frac{\beta}{\eps}, \frac{2\beta}{\eps}]$.

%We fix $q^*: = (n + \frac{1}{\eps})^{-O(1)}$. Let $m := 4\beta/\log \frac{4\beta^2}{\eps q^*}$ rounded up to be next power of $2$, let $g:= 36\log^2\frac{4\beta^2}{\eps q^*}$ rounded up to the next power of $2$, and let $r := \lceil \log \frac{4\beta^2}{\eps q^*} \rceil$. 
Now we are ready to prove Lemma~\ref{lem:mod-color-coding}.

\begin{proof}[Proof of Lemma~\ref{lem:mod-color-coding}]
    We partition $X$ into $\{X^j_{1,1}, \ldots, X^j_{m, g}\}_{j \in [1, r]}$ by Algorithm~\ref{alg:mod-color-coding} as $\mathtt{ModifiedColorCoding}$ $(X, 2\beta, \frac{q^*\eps}{2\beta})$.

    We first prove that for any subset $Y\subseteq X$ with $|Y| \leq 2\beta$, with probability at least $1 - \frac{q^*\eps}{2\beta}$, 
    \[
        \Sigma(Y) \in \bigcup_{j=1}^r S^j_1 +  \cdots + \bigcup_{j=1}^r S^j_{m}.
    \]
    
    We distinguish into two cases. 
    
    If $|X|\leq mg/2$, $S_i^j$'s are the same for different $j$, whereas $\bigcup_{j=1}^r S^j_i = S^1_i$ for all $i\in[1,m]$. So
    \[\bigcup_{j=1}^r S^j_1 +  \cdots + \bigcup_{j=1}^r S^j_m = S^1_1 +  \cdots + S^1_{m} = (X^1_{1,1}\cup \{0\}) + \cdots + (X^1_{m,g} \cup \{0\}).\]
    As each subset has at most one element, it is obvious that for any $Y\subseteq X$, 
    \[
        \Sigma(Y) \in (X^1_{1,1}\cup \{0\}) + \cdots + (X^1_{m,g} \cup \{0\}).
    \]
    
    If $|X|> mg/2$, for any partition $X=X_1\cup X_2$, $Y$ can be partitioned to $Y_1$ and $Y_2$ that $Y_1\subseteq X_1$ and $Y_2\subseteq X_2$. Since $|X_1| = mg/2$, using the same argument as above, we have $\Sigma(Y_1)\in \bigcup_{j=1}^r S^j_1 +  \cdots + \bigcup_{j=1}^r S^j_{m/2}$. Note that $|Y_2|\leq |Y|\leq 2\beta$. According to Lemma~\ref{lem:original-color-coding}, with probability at least $1 - \frac{q^*\eps}{2\beta}$, $\Sigma(Y_2)\in \bigcup_{j=1}^r S^j_{m/2+1} +  \cdots + \bigcup_{j=1}^r S^j_{m}$. So with probability at least $1 - \frac{q^*\eps}{2\beta}$, 
    \[
    \Sigma(Y) = \Sigma(Y_1)+\Sigma(Y_2)\in \bigcup_{j=1}^r S^j_1 +  \cdots + \bigcup_{j=1}^r S^j_{m/2}+\bigcup_{j=1}^r S^j_{m/2+1} +  \cdots + \bigcup_{j=1}^r S^j_{m}.
    \]
    Note that $|\S_X[\frac{\beta}{\eps}, \frac{2\beta}{\eps}]|\leq \frac{\beta}{\eps}+1$. For each $s\in \S_X[\frac{\beta}{\eps}, \frac{2\beta}{\eps}]$, %we just need one subset $Y\subseteq X$ that $\Sigma(Y)= s$. 
    we have shown that the probability that $s=\Sigma(Y)$ for some $Y$ but the event $s\not\in \bigcup_{j=1}^r S^j_1 +  \cdots  + \bigcup_{j=1}^r S^j_{m}$ occurs is at most $\frac{q^*\eps}{2\beta}$. So with probability at least $1 - \frac{q^*\eps}{2\beta}\cdot(\frac{\beta}{\eps}+1) \geq 1-q^*$,
    \[
    \S_X[\frac{\beta}{\eps}, \frac{2\beta}{\eps}] \subseteq (\bigcup_{j=1}^r S^j_1 +  \cdots + \bigcup_{j=1}^r S^j_{m}).
    \]
    It is easy to see that $\S_X[\frac{\beta}{\eps}, \frac{2\beta}{\eps}] \supseteq (\bigcup_{j=1}^r S^j_1 +  \cdots + \bigcup_{j=1}^r S^j_{m})$. So with probability at least $1-q^*$,
    \[
    \S_X[\frac{\beta}{\eps}, \frac{2\beta}{\eps}] = (\bigcup_{j=1}^r S^j_1 +  \cdots + \bigcup_{j=1}^r S^j_{m}).
    \]

    Next, we show that for any $j\in [1,r]$,
    \[
        \max(X^j_{1,1}) + \cdots + \max(X^j_{1,g}) + \cdots +\max(X^j_{m,1}) + \cdots + \max(X^j_{m, g}) \geq \frac{4\beta}{\eps}.
    \]   
    If $|X|\leq mg/2$, as each subset has at most one element,
    \[
    \max(X^j_{1,1}) + \cdots +  \max(X^j_{m, g}) = \Sigma(X)\geq \frac{4\beta}{\eps}.
    \]
    If $|X|> mg/2$, as $mg\geq 144\beta$ and for any $x\in X$ we have that $x\geq \frac{1}{\eps}$, it holds that
    \[
    \max(X^j_{1,1}) + \cdots +  \max(X^j_{m, g}) \geq \max(X^j_{1,1}) + \cdots +  \max(X^j_{m/2, g}) = \Sigma(X_1)\geq \frac{mg}{2}\cdot\frac{1}{\eps}\geq \frac{4\beta}{\eps}.
    \]

    Finally, we analyze the running time. We partition $X$ into $\widetilde{O}(\frac{1}{\eps})$ subsets, and repeat the process of partitioning each subset into $\mathrm{polylog}(n, \frac{1}{\eps})$ subsets $\mathrm{polylog}(n, \frac{1}{\eps})$ times. So the running time is $\widetilde{O}(n+\frac{1}{\eps})$.
\end{proof}

\section{Estimating Sumset Size}\label{app:est-sumset-size}
\lemestsumsetsize*
\begin{proof}
    Suppose $|A|, |B| < k$ since otherwise it is trivial.  Let $\tau$ be any positive number. For any integer set $Z$, define $Z \bmod k = \{z \bmod k : z\in Z\}$. It is easy to verify that 
    \begin{eqnarray}
           |(A \bmod \tau) + (B \bmod \tau)| \leq 2|(A+B) \bmod \tau| \leq 2|A + B|. \label{eq:mod}	
    \end{eqnarray}

    For $i = 1, \ldots, \lceil \log u\rceil$, we compute $(A \bmod 2^j) + (B \bmod 2^j)$ via Lemma~\ref{lem:sparse-fft}. If $(A \bmod 2^j) + (B \bmod 2^j) \geq 2k$ for some $j$, then we immediately stop and conclude that $|A + B| \geq k$ by Eq~\eqref{eq:mod}	. Otherwise, we will reach $j = \lceil \log u\rceil$, and obtain $A+B$ in this round.  Every round takes $O(k\log^5 u\,\mathrm{polyloglog}\,u)$ time except for the last round.  Let $j^*$ be the last round.  Since we do not stop at round $j^*-1$, 
    \[
        |(A \bmod 2^{j^*-1}) + (B \bmod 2^{j^*-1})| < 2k.
    \]
    We make the following claim.
    \begin{claim}\label{claim:mod-sumset}
    	Let A and B be two sets of integers. For any positive integer $k$, 
    	\[|(A \bmod 2k) + (B \bmod 2k)|\leq 3|(A \bmod k) + (B \bmod k)|.\]
    \end{claim}
   Given Claim~\ref{claim:mod-sumset}, we conclude that the last round also takes $O(k\log^5 u\,\mathrm{polyloglog}\,u)$ time since
    \[
        |(A \bmod 2^{j^*}) + (B \bmod 2^{j^*})| \leq 3|(A \bmod 2^{j^*-1}) + (B \bmod 2^{j^*-1})| < 6k.
    \]
    Therefore, the total running time is $O(k\log^6 u\,\mathrm{polyloglog}\,u)$.
It remains to prove Claim~\ref{claim:mod-sumset}. 
\begin{proof}[Proof of Claim~\ref{claim:mod-sumset}]
    Let $C = (A \bmod 2k) + (B \bmod 2k)$, $C' = (A \bmod k) + (B \bmod k)$.
We prove that $C\subseteq \{c', c'+k, c'+2k:c'\in C'\}$. So $|C|\leq 3|C'|$.

For any $c\in C$, we can write $c = a+b$ for some $a \in (A \bmod 2k)$ and $b\in (B \bmod 2k)$. We have $a' = a \bmod k \in (A\bmod k)$ and $b' = b \bmod k \in (B\bmod k)$. Let $c' = a'+b'\in C'$. By definition, it is immediate that $c'\in\{c, c-k. c-2k\}$ and therefore $c\in \{c', c'+k, c'+2k:c'\in C'\}$.
\end{proof}
\end{proof}

\bibliographystyle{alphaurl}
\bibliography{main}
\end{document}